\PassOptionsToPackage{unicode}{hyperref}
\PassOptionsToPackage{hyphens}{url}
\PassOptionsToPackage{dvipsnames,svgnames,x11names}{xcolor}
\documentclass[
  12pt]{article}

\usepackage{amsmath,amssymb}
\usepackage{iftex}
\ifPDFTeX
  \usepackage[T1]{fontenc}
  \usepackage[utf8]{inputenc}
  \usepackage{textcomp} 
\else 
  \usepackage{unicode-math}
  \defaultfontfeatures{Scale=MatchLowercase}
  \defaultfontfeatures[\rmfamily]{Ligatures=TeX,Scale=1}
\fi
\usepackage{lmodern}
\ifPDFTeX\else  
\fi
\IfFileExists{upquote.sty}{\usepackage{upquote}}{}
\IfFileExists{microtype.sty}{
  \usepackage[]{microtype}
  \UseMicrotypeSet[protrusion]{basicmath} 
}{}
\makeatletter
\@ifundefined{KOMAClassName}{
  \IfFileExists{parskip.sty}{%
    \usepackage{parskip}
  }{
    \setlength{\parindent}{0pt}
    \setlength{\parskip}{6pt plus 2pt minus 1pt}}
}{
  \KOMAoptions{parskip=half}}
\makeatother
\usepackage{xcolor}
\makeatletter
\ifx\paragraph\undefined\else
  \let\oldparagraph\paragraph
  \renewcommand{\paragraph}{
    \@ifstar
      \xxxParagraphStar
      \xxxParagraphNoStar
  }
  \newcommand{\xxxParagraphStar}[1]{\oldparagraph*{#1}\mbox{}}
  \newcommand{\xxxParagraphNoStar}[1]{\oldparagraph{#1}\mbox{}}
\fi
\ifx\subparagraph\undefined\else
  \let\oldsubparagraph\subparagraph
  \renewcommand{\subparagraph}{
    \@ifstar
      \xxxSubParagraphStar
      \xxxSubParagraphNoStar
  }
  \newcommand{\xxxSubParagraphStar}[1]{\oldsubparagraph*{#1}\mbox{}}
  \newcommand{\xxxSubParagraphNoStar}[1]{\oldsubparagraph{#1}\mbox{}}
\fi
\makeatother

\usepackage{longtable,booktabs,array}
\usepackage{calc} 
\usepackage{etoolbox}
\makeatletter
\patchcmd\longtable{\par}{\if@noskipsec\mbox{}\fi\par}{}{}
\makeatother
\IfFileExists{footnotehyper.sty}{\usepackage{footnotehyper}}{\usepackage{footnote}}
\makesavenoteenv{longtable}
\usepackage{graphicx}
\makeatletter
\def\maxwidth{\ifdim\Gin@nat@width>\linewidth\linewidth\else\Gin@nat@width\fi}
\def\maxheight{\ifdim\Gin@nat@height>\textheight\textheight\else\Gin@nat@height\fi}
\makeatother

\makeatletter
\def\fps@figure{htbp}
\makeatother

\makeatletter
\@ifpackageloaded{caption}{}{\usepackage{caption}}
\AtBeginDocument{%
\ifdefined\contentsname
  \renewcommand*\contentsname{Table of contents}
\else
  \newcommand\contentsname{Table of contents}
\fi
\ifdefined\listfigurename
  \renewcommand*\listfigurename{List of Figures}
\else
  \newcommand\listfigurename{List of Figures}
\fi
\ifdefined\listtablename
  \renewcommand*\listtablename{List of Tables}
\else
  \newcommand\listtablename{List of Tables}
\fi
\ifdefined\figurename
  \renewcommand*\figurename{Figure}
\else
  \newcommand\figurename{Figure}
\fi
\ifdefined\tablename
  \renewcommand*\tablename{Table}
\else
  \newcommand\tablename{Table}
\fi
}
\@ifpackageloaded{float}{}{\usepackage{float}}
\floatstyle{ruled}
\@ifundefined{c@chapter}{\newfloat{codelisting}{h}{lop}}{\newfloat{codelisting}{h}{lop}[chapter]}
\floatname{codelisting}{Listing}

\makeatother
\makeatletter
\@ifpackageloaded{caption}{}{\usepackage{caption}}
\@ifpackageloaded{subcaption}{}{\usepackage{subcaption}}
\makeatother

\ifLuaTeX
  \usepackage{selnolig}  
\fi
\usepackage[]{natbib}
\usepackage{bookmark}

\IfFileExists{xurl.sty}{\usepackage{xurl}}{} 
\hypersetup{
  pdftitle={Title},
  pdfauthor={Author 1; Author 2},
  pdfkeywords={3 to 6 keywords, that do not appear in the title},
  colorlinks=true,
  linkcolor={blue},
  filecolor={Maroon},
  citecolor={Blue},
  urlcolor={Blue},
  pdfcreator={LaTeX via pandoc}}

\usepackage{amsthm}
\usepackage{multirow}
\usepackage{mathtools}
\usepackage{bbm}
\usepackage[shortlabels]{enumitem}
\usepackage{tikz}

\makeatother

\newcommand{\indep}{\perp \!\!\! \perp}
\def \bbeta {\boldsymbol\beta}

\def \bgamma {\boldsymbol\gamma}

\def \btheta {\boldsymbol\theta}

\def \bpsi {\boldsymbol\psi}

\def \bLambda {\boldsymbol\Lambda}
\def \blambda {\boldsymbol\lambda}

\def \bX {\boldsymbol{X}}

\def \bx {\boldsymbol{x}}

\makeatother

\newtheorem{assumption}{Assumption}
\AfterEndEnvironment{assumption}{\noindent\ignorespaces}
\newtheorem{proposition}{Proposition}
\AfterEndEnvironment{proposition}{\noindent\ignorespaces}

\AfterEndEnvironment{theorem}{\noindent\ignorespaces}

\AfterEndEnvironment{corollary}{\noindent\ignorespaces}

\AfterEndEnvironment{condition}{\noindent\ignorespaces}
\newtheorem{lemma}{Lemma}
\AfterEndEnvironment{lemma}{\noindent\ignorespaces}

\usepackage{xr}

\begin{document}

\def\spacingset#1{\renewcommand{\baselinestretch}%
{#1}\small\normalsize} \spacingset{1}

\title{\bf Causal effects on non-terminal event time with application to antibiotic usage and future resistance}

\author{
Tamir Zehavi${}^{1}$,
Uri Obolski${}^{2}$,
Michal Chowers${}^{3,4}$,
Daniel Nevo${}^{1}$\thanks{danielnevo@tauex.tau.ac.il. TZ and DN were supported by a grant from the Tel Aviv University Center for AI and Data Science (TAD) and by ISF grant 827/21. UO was supported by ISF grant 1286/21.}
}

\date{
\small
${}^1$Department of Statistics and Operations Research, Faculty of Exact Sciences, Tel Aviv University, Tel Aviv, Israel\\
${}^2$Department of Epidemiology and Preventive Medicine, School of Public Health, Faculty of Medical and Health Sciences, Tel Aviv University, Tel Aviv, Israel\\
${}^3$Gray Faculty of Medical and Health Sciences, Tel Aviv University\\
${}^4$Meir Medical Center, Kfar Saba, Israel
}

\maketitle

\bigskip
\begin{abstract}
Comparing future antibiotic resistance levels resulting from different antibiotic treatments is challenging because some patients may survive only under one of the antibiotic treatments. We embed this problem within a semi-competing risks approach to study the causal effect on resistant infection, treated as a non-terminal event time. We argue that existing principal stratification estimands for such problems exclude patients for whom a causal effect is well-defined and is of clinical interest. Therefore, we present a new principal stratum, the infected-or-survivors (ios). The ios is the subpopulation of patients who would have survived  or been infected under both antibiotic treatments.
This subpopulation is more inclusive than previously defined subpopulations. We target the causal effect among these patients, which we term the feasible-infection causal effect (FICE).
We develop large-sample bounds under novel assumptions, and discuss the plausibility of these assumptions in our application. As an alternative, we derive FICE identification using two illness-death models with a bivariate frailty random variable. These two models are connected by a cross-world correlation parameter. Estimation is performed by an expectation-maximization algorithm followed by a Monte Carlo procedure. We apply our methods to detailed clinical data obtained from a hospital setting.
\end{abstract}


\newpage
\spacingset{1.8} 

\section{Introduction}
\label{Sec:intro}

Antibiotic resistance, the ability of bacteria to survive in an environment containing antibiotic drugs, is a major public-health threat,
and has been linked to approximately $4.95$ million annual deaths 
\citep{murray2022global, okeke2024scope}. 
The increasing rates of antibiotic resistance have raised substantial concerns about the future effectiveness of antibiotic treatments \citep{ventola2015antibiotic, okeke2024scope}.
Prior antibiotic usage has been repeatedly correlated with increasing antibiotic resistance rates \citep{ bell2014systematic,  gladstone2021emergence, baraz2023time}, presumably as a result of evolutionary forces selecting for antibiotic resistant bacteria.  
As these findings might be confounded by many factors, causal effects of antibiotic treatment on subsequent antibiotic resistance are generally understudied. 

The WHO AWaRe (Access, Watch, Reserve) antibiotic book \citep{WHO2017} offers recommendations on antibiotic selection, dosing, route of administration, and treatment duration for common clinical infections in both children and adults. \textit{Access} antibiotics have a narrow spectrum of activity, and they generally convey low future antibiotic resistance potential. \textit{Watch} antibiotics are broader-spectrum antibiotics, generally leading to higher future resistance, and are recommended as first-choice options for patients with more severe clinical presentations or for infections where the causative pathogens are more likely to be resistant to Access antibiotics. Finally, \textit{Reserve} antibiotics are those meant to be mainly used as a last-resort antibiotics for severe infections that are resistant to multiple antibiotics. While the consequences of antibiotic choice on future resistance is widely acknowledged, 
only a few studies formally study this causal question \citep{chowers2022estimating, saciuk2025penicillin}. Further methodological advances are essential to enable quantitative estimation of the causal effect of prescribing antibiotics so it is incorporated into clinical decision-making.

The comparison between antibiotic treatments with respect to future antibiotic resistance is complicated by the fact that some patients may only survive under one of the antibiotic treatments.
We embed this problem within the semi-competing risks (SCR) data structure, also termed the illness-death model
\citep{fix1951simple,xu2010statistical}.
SCR data arise when studying times to two events, a non-terminal event and a terminal event. The terminal event can occur after the non-terminal event, but not vice versa.
In our motivating example, death is a terminal event that precludes the occurrence of a future resistant bacterial infection. However, death is possible after a resistant bacterial infection, sometimes as a direct consequence of the newly acquired infection. When studying treatment effects on SCR data, challenges arise in defining and interpreting standard causal effects \citep{nevo2022causal}, similarly to the closely-related truncation by death problem \citep{zhang2003estimation, ding2017principal, zehavi2023matching,tong2025semiparametric}.

Several regression-based models have been proposed for SCR data, aiming to extract different information and tackle various challenges \citep{peng2007regression,hsieh2012regression,lee2015bayesian, nevo2022modeling,gorfine2021marginalized,kats2023accelerated}. These approaches directly model the realized event times distribution, and are not built upon a non-parametric assumption-free definition of causal effects. 

The causal effect of different antibiotic treatments on time to resistant infection is challenging to define. A simple comparison between the cumulative incidence of resistant infections up to a given time point might be misleading due to the treatment effect on the terminal event. For instance, an alleged lower resistance rate might be observed for a certain antibiotic treatment, even if all confounders are properly accounted for, simply due to higher death rates under this treatment.

Recently, approaches anchored within the causal inference paradigm have been proposed \citep{comment2025survivor,xu2020bayesian,nevo2022causal,buhler2023multistate,deng2024direct}. A common theme across this research is the reliance on principal stratification estimands \citep{frangakis2002principal}.  We review several principal stratification approaches in the context of our motivating study in Section \ref{Sec:Setup}. These estimands are based on the notion that within a subset of the population defined by potential event times, causal effects are well-defined.
However, as these subsets are typically not identifiable from the observed data, it is crucial to assess whether the subset is meaningful. We contend that these estimands do not conclusively capture the causal effect of interest, as they overlook patients among which a well-defined causal effect exists and, importantly, is of a clinical interest. 

Therefore, we focus on a different, previously unexplored, principal stratum:
the subpopulation of patients who would have been infected or would have survived within one year after treatment initiation under both antibiotic treatments. The one year mark is chosen based on subject-matter considerations, though other time points can also be considered.
We then define the feasible-infection causal effect (FICE) as the causal effect among this subpopulation. We argue that in applications with SCR data, the FICE provides a meaningful causal effect to address questions about the causal effect on the non-terminal event. In the comparison of Access and Watch antibiotics, the FICE is defined among the subpopulation of patients for which the occurrence of a resistant infection is possible under both antibiotic treatments within the specified time period. 

We consider two overarching identification approaches for the FICE. 
First, we develop a partial identification approach in the form of large-sample bounds under new order-preservation (ORP) assumptions. We present bounds under different assumption combinations and discuss the plausibility of these assumptions in our data. 
As an alternative, we derive FICE identification as a function of an unknown cross-world sensitivity parameter.
Namely, we connect the counterfactual event times using a bivariate frailty random variable, where the cross-world parameter is the correlation between the two frailty variables.
We consider estimation via an expectation-maximization algorithm (EM) \citep{dempster1977maximum}, followed by Monte Carlo approximations. 

The remainder of the paper is organized as follows.
In Section \ref{Sec:data}, we describe the motivating dataset.
In Section \ref{Sec:Setup}, we review and compare existing approaches and estimands in the context of studying causal effects of antibiotic usage, and introduce our proposed estimand. In Section \ref{Sec:Identification}, we present the two identification approaches and discuss the plausibility of the identification assumptions. In Section \ref{Sec:estimation}, we describe the estimation methods for both proposed approaches. In Section \ref{Sec:Application}, we apply our approach to the motivating dataset. In Section \ref{Sec:Discussion}, we offer concluding remarks.

\section{Motivating  dataset}
\label{Sec:data}
 
The data were obtained from Meir Medical Center, a secondary university-affiliated hospital serving a population of approximately 600,000 people. 
The dataset included all 38,791 adult patients who had a positive bacterial culture based on samples taken from their blood, urine, sputum, or wounds, between 1 January 2015 and 22 October 2022. Each culture was tested for susceptibility to several antibiotic types, 
and the dataset included all resistance results. 

The dataset contains hospitalization records, including length of stay, antibiotic treatments and culture results, along with demographics and comorbidities, and date of death. 
We define the \textit{baseline}  %
to be the first instance when a patient received treatment of either amoxicillin-clavulanate or cefuroxime, two commonly prescribed antibiotics in Meir Medical Center, belonging to the Access and Watch classes, respectively.
To enhance clarity, we refer to them as \textit{Access} (amoxicillin-clavulanate) and \textit{Watch} (cefuroxime) antibiotics from now on. We remark that the Access and Watch antibiotic groups each contains many more diverse antibiotics, and our comparison is between amoxicillin-clavulanate and cefuroxime only. 
The treatment as defined in this study is the antibiotic assigned to the patient at baseline. The antibiotic treatment was not randomized and may depend on confounders such as age, medical condition, previous antibiotic usage, and more. 
Table \ref{Tab:table1} describes the distribution of prominent variables among the Watch group (13,959 patients) and the Access group (5,855 patients), categorized by various forms of demographic and clinical information. Patients in the Watch group were generally older; more likely to suffer from several medical conditions, such as dementia 
and chronic renal failure;
and more frequently resided in a residential institution prior to admission.
These characteristics of the data align with clinical standards, as physicians tend to prescribe Watch antibiotics to higher-risk patients.

We define the outcome to be time to ceftazidime-resistant infection, up to one year after baseline. We choose ceftazidime as it can be used as a marker for extended-spectrum beta-lactamase, which is a type of antibiotic resistance mechanism with important consequences for hospital infections \citep{chowers2022estimating}.

We focus on the one-year mark because the association between antibiotic treatment and resistance decreases over time, and is negligible one year after treatment \citep{baraz2023time}. A patient is considered to have acquired a resistant infection if he or she had at least one ceftazidime-resistant culture result between five days and one year after baseline treatment. 

Within one year after baseline, the unadjusted resistant-infection rates among those treated with Access and Watch antibiotics were $7.3\%$ and $9.9\%$, respectively. The corresponding unadjusted one-year death rates were $17.2\%$ and $21.9\%$. At face value, these results might look puzzling, because Watch antibiotics are expected to be more effective in clearing out the bacteria causing the current infection, and consequently should prevent death. However, as discussed above, this is presumably due to confounding, because Watch antibiotics are typically administered to higher-risk patients. After adjusting for potential confounders using a Cox model (Table \ref{Tab:PS_model}), the corresponding one-year death rates
were $25.2\%$ (CI95\%: $23.4\%$, $26.7\%$) in the Access group and $19.9\%$ (CI95\%: $19.2\%$,  $20.6\%$ ) in the Watch group.
The resulting estimated risk 
difference and risk ratio were $-5.3\%$ (CI95\%: $-6.9\%$, $-3.8\%$) and $0.79$ (CI95\%: $0.74$, $0.84$), respectively.

It should be noted that the positivity assumption may be violated in this dataset.  Figure \ref{Fig:Estimated propensity score distribution} shows the distribution of the estimated \textit{propensity score} (PS), the probability of being treated with the Watch antibiotic conditionally on covariates, among each treatment group. For low estimated PS values, there was almost no overlap between the two PS distributions. 
To mitigate this issue, we focus on the causal effect among those treated with the Access antibiotic, as in \cite{chowers2022estimating}. To this end, we use matching \citep{stuart2010matching} 
to adjust for pre-treatment differences in observed covariates between the antibiotic groups. 
Table \ref{Tab:table1} presents descriptive statistics within the original and matched datasets. After matching, covariate balance improved substantially, as can be seen from the notable decrease in standardized mean differences between the two antibiotic groups. 

Table \ref{Tab:events_propoortions} presents the unadjusted event proportions among the $4,143$  matched pairs. Resistant infection rates were relatively similar in the full and matched datasets. In line with the results of the adjusted Cox model in the full dataset, death rates in the matched dataset among the Access antibiotic group were higher than those of the Watch antibiotic group.

To focus on the newly proposed causal estimand and assumptions, we temporarily ignore that the matching implies that the targeted effect is defined among the Access group. We return to this point and describe the matching procedure in Section 
\ref{Sec:Application}. 
\begin{table}
\caption{Event proportions in the amoxicillin-clavulanate (Access) and the cefuroxime (Watch) treatment groups within the matched dataset.} 
\label{Tab:events_propoortions}
\centering
\fbox{
\begin{tabular}{l|cc}
   & Access & Watch \\ 
 \hline 
Sample size & 4,143 & 4,143 \\
Resistant infection & 288 (7\%) & 413 (10.0\%) \\
Infection-free death & 746 (18.0\%) & 616 (14.9\%) \\
Death after infection & 113 (2.7\%) & 147 (3.5\%) \\
\end{tabular}}
\end{table}

\section{Setup and causal estimands}
\label{Sec:Setup}

\subsection{Notations}
\label{Notations}

Let $A$ be the antibiotic treatment, where $A=0$ corresponds to the Access antibiotic and $A=1$ to the Watch antibiotic. Using the potential outcomes framework 
\citep{rubin1974estimating},
for each unit $i=1,...,n$ in the sample, let $T_{i1}(a)$ and $T_{i2}(a)$ be the potential times to resistant infection and death, respectively,
had the antibiotic treatment been set to $A=a$. This formulation assumes there is no interference between units \citep{cox1958planning}. The units are also assumed to be independent and identically distributed, so we omit the $i$ index when it improves clarity.
We set $T_{i1}(a) = \infty$
whenever a resistant infection would not occur within the one year period under treatment level $A=a$. 

Let $S_i(a)=\mathbbm{1}{ \{T_{i2}(a) > 1\} }$ and $I_i(a)=\mathbbm{1}{ \{T_{i1}(a) \le 1\} }$ be the survival and the resistant bacterial infection indicators, respectively, 
within one year under treatment level $a$. Notice that $I_i(a)=1$ is equivalent to $T_{i1}(a) \le T_{i2}(a)$. We term the strata  $\{i: S_{i}(0) = 1,  S_{i}(1) = 1\}$, and 
$\{i: I_{i}(0) = 1,  I_{i}(1) = 1\}$ the \textit{always-survivors} ($as$) and the \textit{always-infected} ($ai$), respectively,
with corresponding population strata proportions $\pi_{as}$ and $\pi_{ai}$.
Other subpopulations, defined by post-treatment periods other than one year, might also be considered. For instance,   
\cite{xu2020bayesian} and \cite{comment2025survivor} have defined time-varying always-survivors subpopulations.
Time-varying subpopulations could be considered in our framework by redefining $S_i(a,r)=\mathbbm{1}{ \{T_{i2}(a) > r\} }$, and $I_i(a,r)=\mathbbm{1}{ \{T_{i1}(a) \le r\} }$ and then studying different subpopulations as a function or $r$.  However, as our focus lies in the one-year post-treatment period, we defer the discussion regarding time-varying subpopulations to Appendix \ref{AppSec:timevaryingsubpop}.

\subsection{Existing approaches and causal effects}
\label{Subsec:Existing}

We now review several approaches that have been proposed to address similar problems.
First, if we ignore follow-up after resistant infections, our application can be described as a competing risks problem, with the competing events being resistant infection and infection-free death.
The total effect $\Pr[T_1(1) \le t] - \Pr[T_1(0) \le t]$ \citep{young2020causal} contrasts infection rates between the antibiotic groups, without eliminating the competing event, i.e., death. However, the interpretation and applicability of the total effect can be challenging in practice.
Low resistant infection rates 
under a specific antibiotic treatment might be observed due to high death rates under this treatment. That is, the total effect might quantify changes in infection rates due to mortality differences \citep{young2020causal,nevo2021reflection}.
Yet, the total effect quantifies the possibly meaningful amount of added number of hospitalized patients with a resistant infection resulting from the use of the Watch antibiotic. Hence, it might be useful for resource management, e.g., for determining the number of beds and medical staff capacity in relevant departments.
Another possible estimand, the controlled direct effect, contrasts infection rates between two antibiotic treatments under the hypothetical intervention that also prevents the occurrence of death. However, such an intervention remains hypothetical, and is impossible to carry out when the competing event is death, as in our study.

Several authors have utilized principal stratification for SCR data to target causal effects among different subpopulations. 
One approach generalizes the survivor average causal effect (SACE) \citep{zhang2003estimation,ding2017principal,zehavi2023matching}, the effect among those who would have survived regardless of their treatment.
To this end, the SACE can be redefined to target time-to-event outcomes, while also recognizing that the non-terminal event may occur at different times rather than being measured at a single, pre-specified, time point.
In our context, a possible estimand is 
\begin{equation*}
    \text{SACE}(t) = \Pr[T_1(1) \leq t | S(0) = 1,  S(1) = 1] - \Pr[T_1(0) \leq t | S(0) = 1,  S(1) = 1], 
    \label{time-fixed SACE_1}
\end{equation*}
for $t\le1$. The SACE$(t)$ 
compares the resistant infection proportions
until time $t$
among the always-survivors.
\cite{xu2020bayesian} and \cite{comment2025survivor}   
 extended the SACE$(t)$ to acknowledge that the terminal event may also occur at different time points, giving rise to time-varying subpopulations (see Appendix \ref{AppSec:timevaryingsubpop}).

Recently, \cite{nevo2022causal} studied the effect of a specific allele on late-onset Alzheimer’s disease (non-terminal event) and death (terminal event) times. They proposed another family of estimands, 
based on a different stratification.
In terms of our application, this stratification divides the population into four strata according to the order of infection and death times. 
The always-infected causal effect (AICE) is the causal effect on resistant infection time among the \textit{always-infected (ai)}, i.e., those who would have had a resistant infection within one year post-treatment, regardless of the antibiotic type.
Formally, the AICE$(t)$ is defined as
\begin{equation*}
    \text{AICE}(t) = \Pr[T_1(1) \leq t | I(0) = 1,  I(1) = 1] - \Pr[T_1(0) \leq t | I(0) = 1,  I(1) = 1],
    \label{always-infected causal effect}
\end{equation*}
for $t \le 1$.
The AICE$(t)$ might be adequate when the main interest lies in the difference in the non-terminal event proportions, and researchers are unwilling to impose that patients will survive by the end of the study under both treatments.  Within the always infected, one could also study the causal effect on the terminal event time \citep{nevo2022causal}.

\subsection{Newly proposed estimands}
\label{subsec:ProposedApproach}

A key issue when discussing effects based on principal or population stratification strategies is the population for which a causal effect is defined. The stratification enables focusing on a well-defined effect, but if the studied stratum excludes certain subpopulations of interest, this effect might be non-informative.
With this in mind, we argue that in our setting, neither the SACE nor the AICE are sufficient to capture causal effects under SCR data. Specifically, in the comparison of Access and Watch antibiotics and their effect on future resistant infection, these estimands exclude from their subpopulation 
relevant patients. 

Define \textit{patient type} ($pt$) to be an index $pt=1,...,16$ representing the quadruple $\{I(0), S(0), I(1), S(1)\}$. Table \ref{Tab:stratum_POs_values} presents all 16 possible $pt$'s; see also Appendix \ref{Appsec:patient types} for a concrete example of potential event times for each patient type. For example, $pt=1$ stands for patients who would have died without being infected under both antibiotic treatments, i.e., $\{I(0), S(0), I(1), S(1)\}=\{0, 0, 0, 0\}$.

Both the always-survivors and the always-infected exclude patients who would have had an infection and subsequently died within one year under one antibiotic treatment, while surviving free of infection under the other. In Table \ref{Tab:stratum_POs_values}, those patient types are $pt=8,9$.
Clearly, there is a clinically-relevant causal effect among these patients, because one antibiotic treatment causes a resistant infection while the other does not. Furthermore, a resistant infection will lead to death for these patient types, highlighting the importance of appropriate treatment for those patients.

\begin{table}
\caption{\footnotesize Subpopulation membership of each patient type (pt), as determined by the values of $\{I(0), S(0), I(1), S(1)\}$. 
'--' means that the patient type does not belong to any of the subpopulations.
ios: infected-or-survivors, ai: always-infected, as: always-survivors.
For a detailed discussion regarding the ORP-type assumptions and the monotonicity assumption, see Section \ref{subsec:bounds}.}
\label{Tab:stratum_POs_values}
\centering
\fbox{\begin{tabular}{c|c|c|c|c|c|c} 
pt & I(0) & S(0) & I(1) & S(1) & Subpopulation & Excluded by \\[0.05cm] 
 \hline
1 & 0 & 0 & 0 & 0 & --- & \\[0.2cm] 
2 & 1 & 0 & 0 & 0 & --- & ORP, ios-ORP, weak-ORP \\[0.2cm]
3 & 0 & 1 & 0 & 0  & --- & ios-ORP, monotonicity\\[0.2cm]  
4 & 0 & 0 & 1 & 0 & --- & \\[0.2cm]
5 & 0 & 0 & 0 & 1  & --- & \\[0.2cm]
\multirow[c]{2}{*}{6} & \multirow[c]{2}{*}{1} & \multirow[c]{2}{*}{1} & \multirow[c]{2}{*}{0} & \multirow[c]{2}{*}{0} & \multirow[c]{2}{*}{---} & ORP, ios-ORP, weak-ORP \\ [0.2cm]
&  &  &  &  &  &  monotonicity\\
 [0.2cm]
7 & 0 & 0 & 1 & 1 & --- & \\[0.2cm]  
8 & 1 & 0 & 0 & 1 & $ios$ & ORP\\[0.2cm]
9 & 0 & 1 & 1 & 0 & $ios$ & monotonicity\\[0.2cm] 
10 & 1 &  0 & 1 & 0 & $ai$, $ios$ & \\[0.2cm]
11 & 0 & 1 & 0 & 1 & $as$, $ios$ & \\[0.2cm]  
12 & 1 &  1 & 1 & 0 & $ai$, $ios$ & monotonicity \\[0.2cm]
13 & 1 &  0 & 1 & 1 & $ai$, $ios$ & \\[0.2cm]
14 & 1 & 1 & 0 & 1 & $as$, $ios$ & ORP\\[0.2cm] 
15 & 0 & 1 & 1 & 1 & $as$, $ios$ & \\[0.2cm]  
16 & 1 & 1 & 1 & 1 & $as$, $ai$, $ios$  &
\end{tabular}}
\end{table}

Furthermore, the SACE and the AICE are defined among subpopulations that exclude additional relevant patients. Specifically, the always-survivors stratum excludes patient types $pt=10,12,13$.
Those patients would have been infected during the one year span under both antibiotics. 
Hence, a causal effect can be defined among them, even though under at least one antibiotic, they would have died following infection.
The always-infected stratum includes those patient types, but excludes patient types $pt=11,14,15$.
As these patients would have survived during the one year post-treatment period, a causal effect is well-defined for them.
  
As a further motivation for studying a larger subpopulation, as targeted by the FICE, recall that the one-year ceftazidime-resistant infection rates in the matched data were relatively low, with 
only $10.0\%$ and $7.0\%$
for patients treated with Watch and Access, respectively. As we show in Section \ref{Sec:Application}, these proportions suggest that the proportion of individuals in the always-infected stratum is low. 
The SACE does not suffer from this specific problem in the motivating dataset. If, however, mortality rates had been high, then the proportion of the SACE-targeted population would have been low.
 
To target a more inclusive and relevant subpopulation, we propose to focus on a new principal stratum.
For $a=0,1$,  let $ios(a) = \{i : \{I_{i}(a)=1\} \cup \{S_{i}(a)=1 \}\}$ be the patient subgroup that would have acquired an infection and/or would have survived within one year after antibiotic treatment $A=a$. We propose to focus on the stratum of patients that belong to both $ios(0)$ and $ios(1)$, which we term the \textit{infected-or-survivors} ($ios$). That is, $ios = \{i : i \in ios(0) \cap ios(1)\}$. We also denote $\pi_{ios}$ for the population proportion of this stratum.
As can be seen from Table \ref{Tab:stratum_POs_values}, the $ios$ subpopulation is more inclusive than both the always-survivors and the always-infected. It contains all patients that belong to at least one of these strata, and in addition, patient types $pt=8,9$ which are neither always-survivors nor always-infected. Looking again at Table \ref{Tab:stratum_POs_values}, we see that the $ios$ stratum covers patient types 8--16. A common feature of the excluded patient types is that for at least one antibiotic treatment $a$, these patients would have died without a resistant infection, that is, $I(a)=S(a)=0$. 

Among the ios stratum, we propose a new estimand, the feasible-infection causal effect (FICE), defined by 
\begin{equation*}
    \text{FICE}(t) = \Pr[T_1(1) \leq t \ |\ ios] - \Pr[T_1(0) \leq t | ios],    
\end{equation*}
for $t\le1$. The FICE$(t)$ is the resistant infection risk difference by time $t$ in the ios stratum. Because infection rates are relatively low, we also consider the analogous risk-ratio scale estimands. For instance, the risk ratio FICE$(t)$ is defined by $\Pr[T_1(1) \leq t \ |\ ios]/\Pr[T_1(0) \leq t | ios]$. 
In Appendix \ref{AppSec:timevaryingsubpop},
we discuss causal effects among time-varying subpopulations \citep{comment2025survivor,xu2020bayesian}.

\subsection{Comparison between the estimands}
\label{Sec:simulations}

In this section, we illustrate and compare the characteristics of the various estimands using a synthetic data generating mechanism (DGM). The DGM was based on Markov Cox cause-specific hazard models with patient-level frailty terms for each antibiotic treatment level \citep{xu2010statistical,nevo2022causal}. Sections
\ref{subsec:frailty} and \ref{Sec:estimation}
provide further information on these models. An additional scenario and the corresponding parameter values for each scenario are given in Appendix \ref{AppSec:simulations}.
The parameter values of the scenario presented here were chosen such that infection and mortality rates would roughly reflect the proportions of the events in the motivating dataset. 

We generated datasets with $n=30,000,000$, and included one continuous covariate, Weibull baseline hazards, and a bivariate Gamma frailty (random effect) model, with means one, variances $\theta$, and a parameter $\rho$ indicating the correlation between the frailty random effects under each treatment value. Loosely, $\rho=0$ means that the event times under different treatments are independent (conditionally on covariates) while values away from zero induce cross-world dependence. Large $\theta$ values capture within-world dependence between $T_1(a)$ and $T_2(a)$, and may also affect the cross-world dependence, depending on the value of $\rho$.

Figure \ref{Fig:initial_simulations_estimands_comparison} compares the FICE$(t)$, the SACE$(t)$ and the AICE$(t)$.
We considered several $\rho$ and $\theta$ values. 
As $\rho$ increased, both the FICE$(t)$ and the SACE$(t)$ were smaller for all $t$ values. 
Further, the estimands varied more with $\rho$ at higher $\theta$ values.
The trajectories of FICE$(t)$ and SACE$(t)$ over time exhibited similar patterns, both demonstrating monotone increases. In contrast, trends in the AICE$(t)$ diverged from those of the other estimands.
At $t=1$ the value of the AICE$(t)$ is zero as expected, because all patients in the always-infected stratum experience a resistant infection within one year after treatment.

The similarities in trends and in magnitude between the FICE$(t)$ and the SACE$(t)$, and their distinction from the AICE$(t)$, resulted from the fact the ios stratum contained more always-survivors than always-infected, due to the substantially higher proportion of the former;
The always-infected proportion was approximately $1\%$, whereas the always-survivors proportion ranged from $74\%$ to $80\%$.
The ios proportion exceeded that of always-survivors by roughly $2\%-3\%$.
This phenomenon highlights the pivotal role of strata proportions in shaping the similarity (or the distinction) between the estimands. Consequently, it is essential for researchers to identify and estimate these proportions and to assess membership overlap across strata when examining the behavior of the estimands.

\begin{figure}
\centering
\caption{\footnotesize{Comparison of different estimands in simulated data over one year following antibiotic treatment.
}}
\label{Fig:initial_simulations_estimands_comparison}
\includegraphics[scale=0.45] 
{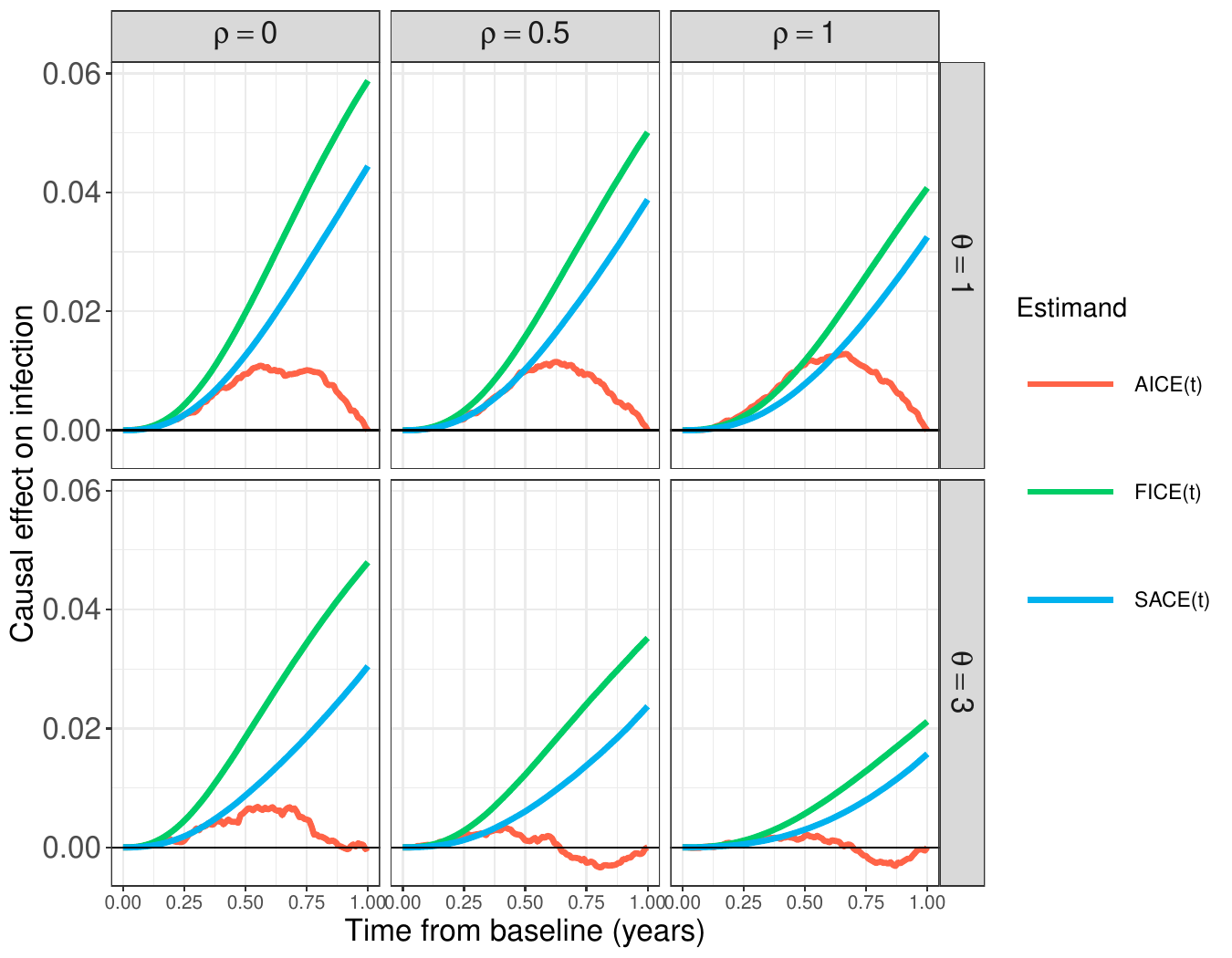}
\end{figure}

\section{Identification of the FICE}
\label{Sec:Identification}

In this section, we describe our two proposed identification approaches. Our first identification approach constructs large-sample bounds for the FICE$(t)$ under relatively weak assumptions. We compare different types of assumptions and the obtained bounds in Section \ref{subsec:bounds}. Our second approach, presented in Section \ref{subsec:frailty},  provides identification up to a sensitivity parameter representing cross-world dependence within a bivariate random effects framework.  
The results we present in this section can be extended and proven for the more general case of time-varying subpopulations. In Appendix \ref{AppSec:Identification},
we discuss the required adjustments for constructing bounds under the general setup, while identification under the frailty assumptions requires no further adjustments.
We also present the results for the risk-ratio scale estimands and all the identification proofs in that appendix.

We assume consistency, so for each patient $i$, the observed outcomes are $T_{i1}=T_{i1}(A_i)$ and $T_{i2}=T_{i2}(A_i)$. In our dataset, follow-up within the first year after treatment assignment is complete for nearly all patients (censoring rate before the one year mark was approximately 8.5\%).
We therefore focus on the case of no right censoring, and account for it in the estimation. The identification strategy can be naturally extended to allow for censoring, similarly to \cite{axelrod2023sensitivity}.  
Let $\widetilde{T}_{i1} = \min(T_{i1},T_{i2},1)$ 
and
$\widetilde{T}_{i2} = \min(T_{i2},1)$  
be  the observed times, 
and let 
$\delta_{i1} = \mathbbm{1}\{ T_{i1} \le 1 \}$ and
$\delta_{i2} = \mathbbm{1}\{ T_{i2} \le 1 \}$
be
indicators for the occurrence of a resistant infection and death, respectively, during follow-up.

When selecting antibiotic treatment, physicians take into account the patient’s clinical status. 
For instance, patients classified as high-risk are more likely to receive the Watch antibiotic.
Therefore, in our motivating problem, we cannot assume that the antibiotic treatment was allocated randomly.
We instead leverage the detailed dataset and assume conditional exchangeability (CE), i.e., $A \indep 
\{T_1(a),T_2(a)\} | \bX$, for $a=0,1$.  Similarly to the identification of the SACE and the AICE, which often requires additional assumptions,  
consistency and CE do not suffice for FICE identification, and additional assumptions are required. 

To summarize, the observed data for each patient $i$
are $(\bX_i , A_i ,\widetilde{T}_{i1}, \delta_{i1}, \widetilde{T}_{i2}, \delta_{i2})$. For any event $\mathcal{Q}$, let $\Psi_\mathcal{Q}(t) = \Pr\big(\{T_1 \le t\} \cup \{T_2 > t\}|\mathcal{Q} \big)$ be the conditional probability of being  infected by time $t$, surviving beyond time $t$, or both. To avoid clutter, we also let $\Psi_\mathcal{Q}=\Psi_\mathcal{Q}(1)$. For $j=1,2$, let 
$F_{j|\mathcal{Q}}(t) = \Pr\big(T_j \le t|\mathcal{Q} \big)$ be the conditional probability of the $j$-th event (infection or death), occurring before time $t$. Finally,  denote $S_{j|\mathcal{Q}}(t) = 1 - F_{j|\mathcal{Q}}(t)$.

\subsection{Bounds under ORP-type assumptions}
\label{subsec:bounds}

We start by presenting our partial identification approach, in the form of large-sample bounds for the FICE. 
To provide motivation for the different assumptions, we revisit our clinical context. Watch antibiotics are aimed at clearing the more severe infections, and consequently, within the same population, they are expected to increase survival probability, compared with Access antibiotics. 
Concurrently, this major benefit is not necessarily devoid of associated 
clinical costs, as Watch antibiotics are also more likely to generate future antibiotic resistance, as a result of ecological and evolutionary forces selecting for antibiotic resistant bacteria. 

Principal stratum effect identification often relies on monotonicity-type assumptions. 
In our notation, the monotonicity assumption is $S_i(0)\le S_i(1)$, for all $i$. That is, all patients who would have survived beyond one year under the Access antibiotic, would also survive under the Watch antibiotic. 
Monotonicity within a short time frame, during the hospital stay, might seem plausible; a baseline life-threatening infection cleared by the Access antibiotic would also be cleared when treated with the more aggressive Watch antibiotic. Monotonicity will exclude, for example, individuals of patient type $pt=3$ (Table \ref{Tab:stratum_POs_values}), which are indeed not expected to be part of the population. Such patients would have survived only under the Access antibiotic ($S(0)=1,S(1)=0$), but their death under the Watch antibiotic is not due to a new resistant infection ($I(0)=0,I(1)=0$).  

However, when considering the wider one-year time window, monotonicity is no longer plausible, because the Watch antibiotic may lead to a subsequent resistant infection, which may lead to death. This would be the case for $pt=9$, which belongs to the $ios$ stratum. Of note is that for the SACE, monotonicity alone does not imply point identification \citep{zhang2003estimation,ding2017principal,zehavi2023matching}.

Therefore, we consider assumptions of a different nature, that also respect the SCR nature of the data. Our approach extends the approach of \cite{nevo2022causal}, who considered the following order preservation assumption (ORP).
\begin{assumption}
\label{assum:ORP}
Order-Preservation (ORP). 
$$
\forall i: \text{If} \; I_i(0)=1, \text{then} \; I_i(1)=1. 
$$
\end{assumption}
In our application, ORP implies that the Watch antibiotic cannot prevent a future resistant infection. We focus on bounds under more subtle assumptions for two reasons. First, as previously explained, patients of $pt=3$ are unlikely to exist in the population, and the ORP assumption does not exclude this patient type. Second, as seen on Table \ref{Tab:stratum_POs_values}, ORP excludes from the population two patient types that belong to the target population ($pt=8,14$). Although it seems plausible to assume away the inclusion of these patient types in the population, we prefer to avoid assumptions on the target subpopulation itself. 
 
We introduce a new assumption, which we term the infected or survivors order-preservation assumption ($ios$-ORP).
\begin{assumption}
\label{assum:ios-ORP}
Infected or Survivors Order-Preservation ($ios$-ORP).
$$
\forall i: \text{If} \; i \in ios(0), \text{then} \; i \in ios(1).
$$
\end{assumption}
\noindent Under the ios-ORP assumption, every patient that belongs to the $ios(0)$ subpopulation, also belongs to the $ios(1)$ subpopulation. Namely, any patient who would have been infected or survived under the Access antibiotic would also have been infected or survived had they received the Watch antibiotic. 

As seen in Table \ref{Tab:stratum_POs_values}, the ios-ORP assumption excludes patient types $pt=2,3,6$. The logic behind excluding $pt=3$ has been previously presented. For $pt=2,6$, a future resistant infection occurs under the Access antibiotic treatment, while under the Watch antibiotic these patients would have died without having a resistant infection. Such scenarios are unlikely, because, as previously explained, the Watch antibiotic is more effective against the baseline infection (so death is less likely), but has higher chances of future infection. 

The ios-ORP assumption can be falsified by the observed data. Under ios-ORP, $\Pr(\{T_1(0)<1 \}\cup \{T_2(0)>1\}) \le  \Pr(\{T_1(1)<1\} \cup \{T_2(1)>1\})$. Therefore, under CE, we would expect that  $E_{\bX}[\Psi_{A=0,\bX}] \le E_{\bX}[\Psi_{A=1,\bX}]$. If this inequality does not hold in the data, then either CE or ios-ORP do not hold. 

The next new assumption we introduce is weak-ORP. It is weaker than the two previous assumptions, as shown in Appendix \ref{AppSec:Identification}. 
\begin{assumption}
\label{assum:weak-ORP}
Weak Order-Preservation (weak-ORP). 
$$
\forall i: \text{If} \;
 I_i(0)=1, \text{then} \;  i \in ios(1).
 $$
\end{assumption}
Table \ref{Tab:stratum_POs_values} shows which patient types are excluded by each ORP assumption. Under all three ORP variants, a patient who would have been infected under the Access antibiotic belongs to the $ios$ subpopulation.
Weak ORP is the weakest assumption (Lemma \ref{lem:orp_assumptions}). All patient types excluded under this assumption 
($pt=2,6$)
are also excluded under each of the other two assumptions (separately).
Under ios-ORP, $pt=3$ is also excluded, while under ORP it is not, but $pt=8,14$ are excluded. 
Importantly, the only patient type that is excluded by both weak-ORP and ios-ORP but not under monotonicity is $pt=2$, which is unlikely to be present in our population of interest.
Unlike the ORP assumption, the 
ios-ORP and weak-ORP assumptions do not exclude patient types that are part of the $ios$ stratum. We establish bounds for the FICE$(t)$ under the latter two ORP-type assumptions. In Appendix \ref{AppSec:Identification}, we also present bounds under consistency and 
CE only, without imposing any of the ORP assumption. 

Starting with ios-ORP, we show in 
Lemma \ref{lem:PIiosIdentORP} that under consistency, CE and 
ios-ORP, the proportion of the $ios$ stratum is identified by 
$\pi_{ios} = E_{\bX}[\Psi_{A=0,\bX}]$. 
The proposition below forms FICE$(t)$ bounds under these assumptions.
\begin{proposition}
\label{Prop:FICE_bounds_ios_ORP}
Under consistency, CE and $ios$-ORP, 
the FICE$(t)$ is bounded by
\begin{align*}
& \mathcal{L}(t) \le \text{FICE($t$)} \le \mathcal{U}(t), 
\end{align*}
where
\begin{align*}
& \mathcal{U}(t) = \min\Bigg\{1, \frac{E_{\bX}[F_{1|A=1,\bX}(t)]} {E_{\bX}[\Psi_{A=0,\bX}]}\Bigg\}  - \frac{E_{\bX}[F_{1|A=0,\bX}(t)]}  {E_{\bX}[\Psi_{A=0,\bX}]}, \\
& \mathcal{L}(t) = \max\Big\{0, 1 - 
\frac{E_{\bX}[S_{1|A=1, \bX}(t)]} {E_{\bX}[\Psi_{A=0,\bX}]}\Bigg\}  - \frac{E_{\bX}[F_{1|A=0,\bX}(t)]}  {E_{\bX}[\Psi_{A=0,\bX}]}. 
\end{align*}
\end{proposition}
Specifically, under ios-ORP, $\Pr(T_1(0) \le t|ios)$
is identified by 
$E_{\bX}[F_{1|A=0,\bX}(t)]/  E_{\bX}[\Psi_{A=0,\bX}]$, while $\Pr(T_1(1) \le t|ios)$ is only partially identified. We observe that as 
$\pi_{ios}$ (identified by $E_{\bX}[\Psi_{A=0,\bX}]$) 
increases, the bounds get tighter. 
Moreover, if 
\begin{equation}
\label{Eq:cond}
1 - \pi_{ios} < E_{\bX}[F_{1|A=1,\bX}(t)] < \pi_{ios},
\end{equation} 
then both of the bounds for $\Pr(T_1(1) \le t|ios)$ are informative. Condition \eqref{Eq:cond} is  satisfied only if $\pi_{ios} > 0.5$.
For all $t$ such that \eqref{Eq:cond} holds, 
the length of the interval  $\mathcal{U}(t) - \mathcal{L}(t)$ is constant and equals to $1/\pi_{ios} - 1$.

We now present bounds under the relaxed ORP version, namely, weak-ORP.
Under weak-ORP, the $ios$ proportion is no longer fully identified. Instead,
Lemma \ref{lem:PIios_weakORP_bounds}
provides bounds for $\pi_{ios}$ under weak-ORP, denoted by $\mathcal{\tilde{U}}_{\pi}$ and $\mathcal{\tilde{L}}_{\pi}$, and use them to derive FICE$(t)$ bounds.
\begin{proposition}
\label{Prop:FICE_bounds_weak_ORP}
Under consistency, CE and weak-ORP, the FICE$(t)$ is bounded by
\begin{align*}
\mathcal{\tilde{L}}(t) \le \: \:  & \text{FICE$(t)$} \le \:  \mathcal{\tilde{U}}(t), 
\end{align*}
where
\begin{equation*}
\mathcal{\tilde{U}}(t) = \dfrac{{\tilde{u}}(t)}{1\{{\tilde{u}}(t) \ge 0\}  \mathcal{\tilde{L}}_{\pi} + 1\{\tilde{{u}}(t) < 0\}  \mathcal{\tilde{U}}_{\pi}} \:, 
\mathcal{\tilde{L}}(t) = \dfrac{\tilde{{l}}(t)}{1\{\tilde{{l}}(t) \ge 0\}  \mathcal{\tilde{U}}_{\pi} + 1\{\tilde{{l}}(t) < 0\}  \mathcal{\tilde{L}}_{\pi}},
\end{equation*}
and
\begin{align*}
& \tilde{u}(t) = \min\Big\{E_{\bX}[F_{1|A=1,\bX}(t)], E_{\bX}[\Psi_{A=0,\bX}]\Big\} - E_{\bX}[F_{1|A=0,\bX}(t)],  \\
& \tilde{l}(t) = \max\Big\{0, E_{\bX}[\Psi_{A=0,\bX}] -E_{\bX}[S_{1|A=1,\bX}(t)]\Big\} - E_{\bX}[F_{1|A=0,\bX}(t)], 
\end{align*}
with $1\{\cdot\}$ being the indicator function. 

\end{proposition}

\subsection{Identification under frailty assumptions}
\label{subsec:frailty}
As an alternative to the partial identification approach detailed above, we also provide a sensitivity analysis for the FICE$(t)$ as a function of a single parameter. Our approach connects two random-effect illness-death models, as originally proposed by \cite{nevo2022causal}.
Each model represents transition probabilities between states $jk \in \{01,02,12\}$ for one treatment group, and the models are connected by a bivariate frailty random variable $\bgamma = (\gamma_0, \gamma_1)$. 
The cause-specific hazard functions that
define each treatment-specific illness-death model are 
\begin{align*}
\begin{split}
\lambda_{01}(t|a,\bX,\bgamma) = \displaystyle{\lim_{h \to 0}} \: \frac{1}{h}\Pr(T_1 \in [t, t+h)|A=a, T_1 \ge t, T_2 \ge t, \bX, \bgamma),  \:\: & t>0 \\[0.5em] 
\lambda_{02}(t|a,\bX,\bgamma) = \displaystyle{\lim_{h \to 0}} \: \frac{1}{h}\Pr(T_2\in [t, t+h)|A=a, T_1 \ge t, T_2 \ge t, \bX,\bgamma),  \:\: & t>0 \\[0.5em]
\lambda_{12}(t|a,t_1,\bX,\bgamma) =  \displaystyle{\lim_{h \to 0}} \: \frac{1}{h}\Pr(T_2 \in [t, t+h)|A=a, T_1 = t_1, T_2 \ge t,  \bX,\bgamma,)\:\: & t>t_1>0.
\end{split}
\end{align*}
for $a=0,1$.

The following frailty assumptions establish the connection between the distributions of the potential event times through the frailty $\bgamma$,
where an unidentifiable parameter, denoted as $\rho$, governs the cross-world correlation.
\begin{assumption}
\label{assump:frail}
There exists a bivariate random variable $\bgamma=(\gamma_0, \gamma_1)$, fulfilling the following conditions: 
\begin{minipage}[t]{\linewidth}
\begin{enumerate}
[(i), 
partopsep=0.1pt,  
leftmargin=-1.4cm, 
rightmargin=4cm
]
\item[]
\item  $A \indep 
  \{T_1(a), T_2(a), \gamma_a\}|\bX$, for $a=0,1$.
\item Given $\bX$ and $\bgamma$, the joint distribution of the potential event times can be separated into the following components
\begin{align*}
    &f(T_1(0), T_2(0), T_1(1), T_2(1)|\bX,\gamma_0, \gamma_1) =
    f(T_1(0), T_2(0)|\bX,\gamma_0)
    f(T_1(1), T_2(1)|\bX,\gamma_1),
\end{align*}
where $f(\cdot)$ is a density function. 
\item $\bgamma \indep \bX$. In words, the frailty variable $\bgamma$ and the covariates $\bX$ are independent. 
\item The hazard function has the following multiplicative form:
$\lambda_{jk}(t|a,\bX,\bgamma) = \gamma_a \lambda'_{jk}(t|a,\bX)$ for $j=0, k=1,2$, and 
$\lambda_{12}(t|t_1,a,\bX,\bgamma) = \gamma_a \lambda'_{12}(t|t_1,a,\bX)$, for $a=0,1$ and for some functions $\lambda'_{jk}$.
\item The probability density function of $\bgamma$, termed $f_{\btheta}(\bgamma)$, is known up to the parameters $\btheta$, and the parameters(s) that govern $f_{\btheta}(\bgamma)$ are identifiable from the observed data.
\end{enumerate}
\end{minipage}
\end{assumption}
\noindent
Part (i) is similar to the CE assumption, where $\gamma_a$ is also included together with the potential event times.
Part (ii) asserts that conditionally on the covariates $\bX$ and the frailty $\bgamma$, the cross-world event time pairs are independent, and that conditionally on the frailty $\gamma_{a}$ and covariates $\bX$, potential event times under treatment $a$ are independent of $\gamma_{1-a}$.
Parts (iv) and (v) pertain to the identification of the observed data distribution (for $T_1<T_2$) and of $f_{\btheta}(\bgamma)$.

The proposition below provides the identification of the FICE$(t)$ and of the $ios$ stratum proportion, as a function of the unidentifiable parameter $\rho$, under the frailty assumptions.

\begin{proposition}
\label{Prop:FICE_identification_frailty}
Under consistency and the frailty assumptions (Assumption \ref{assump:frail}), 
the FICE is identified from the observed data by
\begin{equation*}
\text{FICE$(t)$} = 
\frac{1}{\pi_{ios}(\btheta, \rho)}
\big\{E_{\bX}\big[E_{\bgamma}\big[F_{1|A=1, \bX, \gamma_1}(t) \Psi_{A=0, \bX, \gamma_0} - F_{1|A=0, \bX, \gamma_0}(t) \Psi_{A=1, \bX, \gamma_1} \big]\big] \big\},
\end{equation*}
where 
${\pi_{ios}(\btheta, \rho)} = E_{\bX}\big\{E_{\bgamma}[\Psi_{A=0, \bX, \gamma_0} \Psi_{A=1, \bX, \gamma_1}]\big\}$ 
is the proportion of the $ios$ stratum in the population under the frailty assumptions.
\end{proposition}
\noindent Monte Carlo simulations can be used to approximate these expectations. The SACE$(t)$ can be identified under the frailty assumptions as well.

In our data analysis in Section \ref{Sec:Application}, we assume that for $a=0,1$, $\gamma_a$ is a Gamma frailty variable, with a shape parameter $\theta_a^{-1}$ and a scale parameter $\theta_a$.
The vector $\bgamma$
is then a bivariate Gamma variable, with means one, variances $\theta_0$ and $\theta_1$, and the sensitivity parameter $\rho$ is the correlation between $\gamma_0$ and $\gamma_1$. More generally, researchers can also specify other probability distributions for $\bgamma$.

\section{Semi-parametric model and estimation}
\label{Sec:estimation}

Calculating the large-sample bounds entails estimation of $F_{1|A=a,\bX}(t)$ and $\bpsi_{A=a,\bX}(t)$, for both $a=0,1$. This could be done using any regression-based illness-death model. Frailty random variables can be used to model dependence between $T_1(a)$ and $T_2(a)$, without modeling cross-world independence. We describe one such approach at the end of this section. 
 
For estimation of and inference about the FICE$(t)$ under the frailty assumptions, we will use the following hazard models.
\begin{align}
\begin{split}
\label{conditional hazard models DGM}
\lambda_{01}(t|a,\bX,\bgamma) = \gamma_a \lambda^0_{01}(t|a) \exp(\bX^t\bbeta^a_{01}),  \:\: & t>0 \\[0.5em] 
\lambda_{02}(t|a,\bX,\bgamma) = \gamma_a \lambda^0_{02}(t|a) \exp(\bX^t\bbeta^a_{02}),  \:\: & t>0 \\[0.5em]
\lambda_{12}(t|t_1, a,\bX,\bgamma) = \gamma_a \lambda^0_{12}(t|a) \exp(\bX^t\bbeta^a_{12}), \:\: & t>t_1>0, 
\end{split}
\end{align}
with unspecified baseline hazard functions $\lambda^0_{jk}(t|a)$.
Denote the cumulative
baseline hazard functions as $\Lambda^0_{jk}(t|a) = \int_{0}^{t} \lambda^0_{jk}(u|a) du$.
We define the following quantities: 
$\widetilde{\btheta} = (\theta_0, \theta_1),
\widetilde{\bbeta} = (\bbeta^0, \bbeta^1),
\widetilde{\blambda}_0(t) = (\blambda^0(t), \blambda^1(t))$, and
$\widetilde{\bLambda}_0(t) = (\bLambda^0(t), \bLambda^1(t))$, where, for $a=0,1$,
$\bbeta^a = (\bbeta^a_{01}, \bbeta^a_{02}, \bbeta^a_{12}),
\blambda^a(t) = (\lambda^0_{01}(t|a), \lambda^0_{02}(t|a), \lambda^0_{12}(t|a))$, and
$\bLambda^a(t) = (\Lambda^0_{01}(t|a), \Lambda^0_{02}(t|a), \Lambda^0_{12}(t|a))$. 
Define also 
\begin{align*}
k_i &= H^0_{01}(\widetilde{T}_{i1}|a_i)\exp(\bX_i^t\bbeta^{a_i}_{01}) + H^0_{02}(\widetilde{T}_{i1}|a_i)\exp(\bX_i^t\bbeta^{a_i}_{02})\\[0.5ex]
&+ \big[H^0_{12}(\widetilde{T}_{i2}|a_i) - H^0_{12}(\widetilde{T}_{i1}|a_i)\big]  \exp(\bX_i^t\bbeta^{a_i}_{12})\delta_{i1}.
\end{align*}
In  Appendix \ref{Appsubsec:Estimation} we show that the likelihood function, integrated over the frailty distribution,
is equal to
\begin{align}
\begin{split}
  L&(\widetilde{\btheta}, \widetilde{\bLambda}_0, \widetilde{\blambda}_0, \widetilde{\bbeta})  = \prod_{i=1}^{n} \Big\{\ \big[ \lambda^0_{01}(\widetilde{T}_{i1}|a) \big]^{\delta_{i1}} 
 \big[ \lambda^0_{02}(\widetilde{T}_{i1}|a) \big]^{(1 - \delta_{i1}) \delta_{i2}} 
 \big[ \lambda^0_{12}(T_{i2}|a)  \big]^{\delta_{i1} \delta_{i2}} \\[0.5ex]
  & \exp\big( \delta_{i1}X_i^t\bbeta^{A_i}_{01} +
(1 - \delta_{i1})\delta_{i2}X_i^t \bbeta^{A_i}_{02} + 
+
\delta_{i1} \delta_{i2}X_i^t \bbeta^{A_i}_{12}\big) (-1)^
  {\delta_{i1} + \delta_{i2}} \phi_{a_i}^{(\delta_{i1} + \delta_{i2})}(k_i) \Big\},
  \nonumber
\end{split}
\end{align}
where for $a=0,1, q=0,1,2$, the function $\phi_{a}^{(q)}(k)$ is the $q$-th derivative of 
$E[\exp(-k \gamma_a)]$
with respect to $k$.
 
Estimation of the parameters can be achieved by employing an EM algorithm \citep{dempster1977maximum}, described in Appendix \ref{Appsubsec:Estimation}.
After parameter estimation, we use Monte Carlo simulations to estimate the quantities in the identification result from Proposition \ref{Prop:FICE_identification_frailty}.
Standard errors (SEs) and confidence intervals (CIs) can be estimated using a bootstrap procedure.

Bounds estimation can be based on the same models, while ignoring the correlation parameter $\rho$ as the quantities that establish the bounds do not involve any cross-world terms. Importantly, in this approach, the frailty random variables are used only for estimation purposes, not identification.

\section{Application to the motivating example}
\label{Sec:Application}

We applied our two identification approaches to quantify the causal effect of using a cefuroxime, a Watch antibiotic treatment, compared with amoxicillin-clavulanate, an Access antibiotic treatment, on future ceftazidime-resistant infection.
As discussed in Section \ref{Sec:data}, we first matched each patient in the Access group to a single patient in the Watch group, without replacement.
Matching was based on the Mahalanobis distance with a \textit{caliper} \citep{rubin2000combining,stuart2010matching} of $0.3$ standard deviations on the estimated PS. Both the Mahalanobis distance and the PS model (Section \ref{AppSec:Motivating data}) included all covariates described in Table \ref{Tab:table1}.
Due to the difference in the estimated PS distributions between the two antibiotic groups, only 4,143 out of the $5,855$ ($71\%$) patients who received Access antibiotic had a satisfactory match. 
However, the results did not change substantially when a larger caliper was taken or when matching was done solely on the estimated PS. Due to the matching, the causal effects we study cover the population of patients who are treated with Access antibiotic. 

SEs were estimated
by a pair-level Bootstrap with 200 repetitions.
Wald-type $95\%$ confidence intervals were then calculated.

Starting with the large-sample bounds, the quantities needed for calculating the bounds were estimated by two frailty illness-death models, estimated separately at each treatment group by the EM algorithm (Appendix \ref{Appsubsec:Estimation}). Let $\widehat{E}_{\bX}(\cdot)$ be the empirical mean over the distribution of $\bX$ in the matched dataset.
The estimated mean $\Psi$'s were
$\widehat{E}_{\bX}[\widehat{\Psi}_{A=0,\bX}]= 0.815$ and $\widehat{E}_{\bX}[\widehat{\Psi}_{A=1,\bX}]=0.876$. Because $\widehat{E}_{\bX}[\widehat{\Psi}_{A=0,\bX}]<\widehat{E}_{\bX}[\widehat{\Psi}_{A=1,\bX}]$, the ios-ORP assumption was not falsified by our data.
The estimated ios stratum proportion under ios-ORP was $81.5\%$ (CI95\%: $80.5\%$, $82.4\%$).
The bounds for $\pi_{ios}$ under weak-ORP and without ORP assumptions were identical (Appendix \ref{AppSec:Identification}) and estimated to be $[69.1\%, 81.5\%]$.
 
Figure \ref{Fig:res_ORP_with_diff} presents the estimated FICE$(t)$ bounds. On the difference scale, the one-year estimated FICE bounds without ORP assumptions, under weak-ORP, and under ios-ORP were 
$[-9.7\%, 12.4\%]$, $[-9.7\%, 2.6\%]$ and $[-8.3\%, 2.2\%]$, respectively. 
As one may expect, the stronger the assumption taken, the narrower the bounds became.
After approximately 90 days ($t\approx0.25$), the upper bounds under the ORP assumptions did not vary considerably.
On the risk-ratio scale (Figure \ref{Fig:res_upper_bound_under_weak_ORP_ratio}), the bounds under weak-ORP and ios-ORP coincide analytically (see Appendix \ref{AppSec:Identification}).
The upper bound under weak-ORP decreased from approximately $2.00$ two weeks after baseline to $1.27$ (CI95\%: $1.07$, $1.46$) after one year,
and the lower bound was zero for all $t$, i.e., it was not informative. 
The estimated bounds without ORP-type assumptions were also not informative.

 \begin{figure}
\centering
\caption{\footnotesize{Estimated bounds for the FICE$(t)$ on the difference-scale without ORP-type assumptions (without ORP), and under weak-ORP and ios-ORP.
The lower bound without ORP assumptions was the same as the one under weak-ORP.}}
\label{Fig:res_ORP_with_diff}
\centering
\includegraphics[scale=0.5] 
{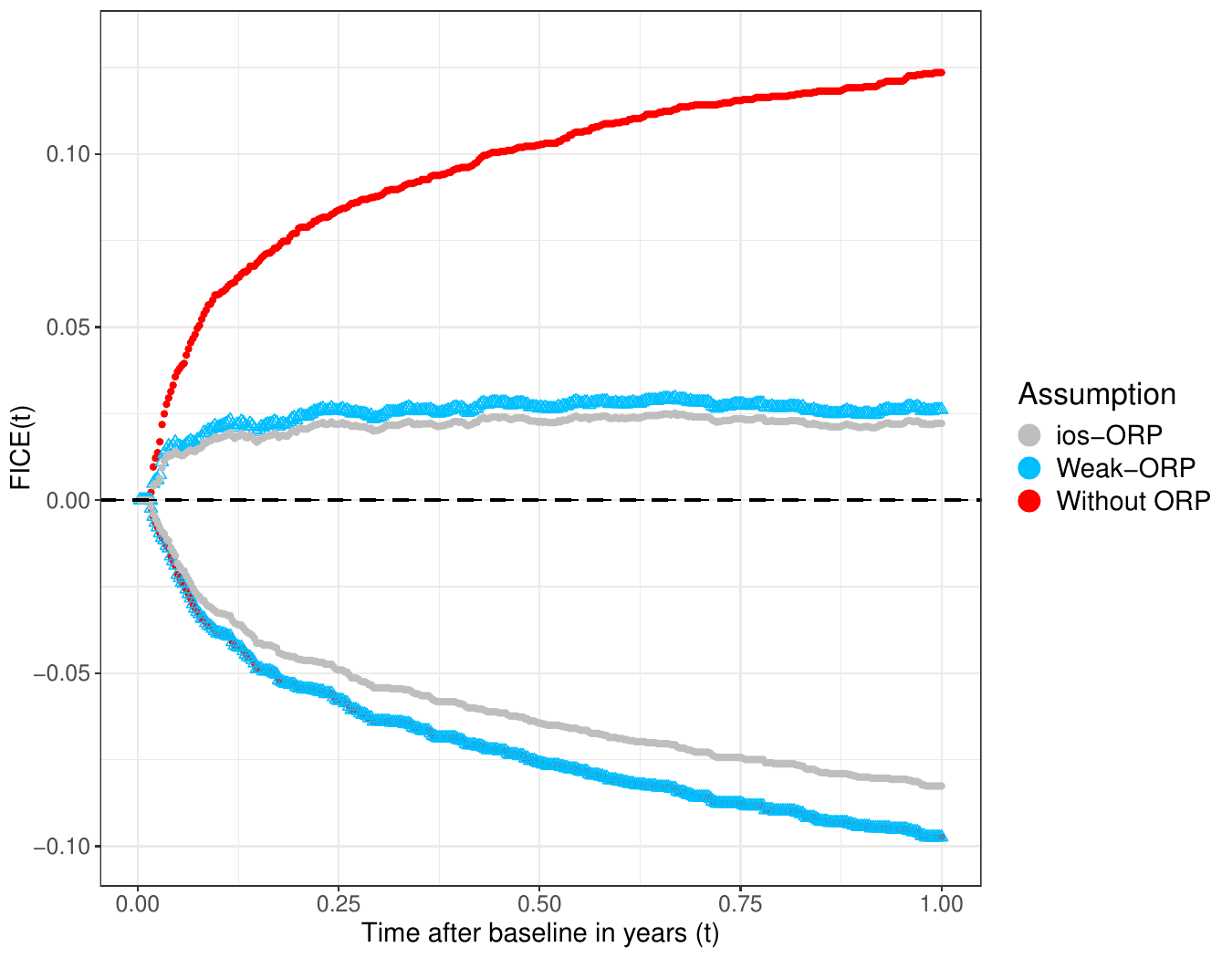}
\end{figure}

Turning to the alternative sensitivity analysis under the frailty assumptions, we employed the EM algorithm for the illness-death models with bivariate Gamma frailty distribution with varying $\rho=Corr(\gamma_0,\gamma_1)$ values.
The obtained estimates were relatively robust to different $\rho$ values. Therefore, we present here the results for $\rho=0$, and
compare the results for other $\rho$ values in Appendix \ref{Appsec: Application}. 
The estimated values of $\theta_0$ and $\theta_1$ were 1.78 (CI95\%: $1.17$, $2.52$) and 0.89 (CI95\%: $0.33$, $1.42$), respectively.
The ios proportion was estimated to be 4.5\% larger than the always-survivors proportion. These estimated proportions were 
$\widehat{\pi}_{ios}=74.0\%$ (CI95\%: $72.9\%$, $75.2\%$) and $\widehat{\pi}_{as}=69.5\%$ (CI95\%: $68.4\%$, $70.7\%$).
Due to the overlap between these subpopulations, the estimated effects within the two strata did not differ substantially. The estimated always-infected proportion was notably low ($\widehat{\pi}_{ai}=0.9\%$) and we therefore do not present results for the AICE.

\begin{figure}
\caption{\footnotesize{Comparison between estimates of the different estimands within one year after baseline under the frailty assumptions with $\rho=0$. 
The upper panel shows the various estimates on the difference scale. The lower panel shows the various estimates on the risk-ratio scale.}}
\label{Fig:res_frailty_all_diff_and_ratio}
\centering
\includegraphics[scale=0.59]{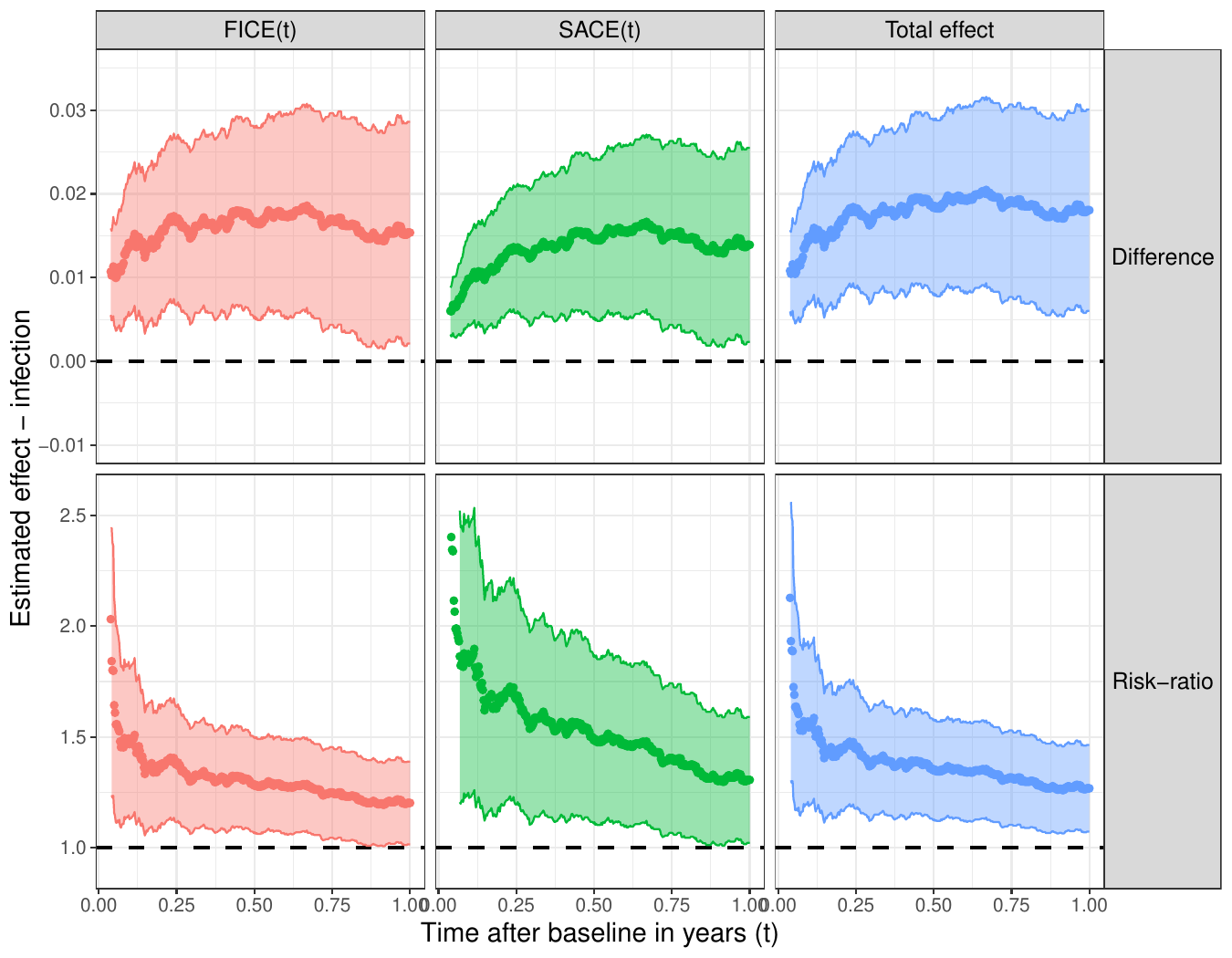}
 \end{figure}

Figure \ref{Fig:res_frailty_all_diff_and_ratio} 
depicts estimated causal effects for $t\le 1$ along with the 95\% confidence intervals.
At $t=1$, the FICE$(t)$ and the SACE$(t)$ were estimated to be $1.5\%$ (CI95\%: $0.2\%$, $2.9\%$) 
and $1.4\%$ (CI95\%: $0.2\%$, $2.6\%$),
respectively. 
The estimated total effect at $t=1$ was $1.8\%$ (CI95\%: $0.6\%$, $3.0\%$). The total effect was about $20\%$ larger than the FICE and the SACE, as the Watch antibiotic presumably prevents death, resulting in higher infection rates under the Watch antibiotic. This finding illustrates the subtlety when interpreting the total effect.

The proportion of the two patient types that would have been infected and subsequently died under one antibiotic treatment but not under the other ($pt=8,9$) was 
$3.9\%$ (CI95\%: $3.4\%$, $4.3\%$).
The effect among them was relatively high, although not significant, amounting to $4.4\%$ (CI95\%: $-8.7\%$, $17.5\%$), highlighting the necessity of including them within the population of interest.
As discussed in Section \ref{subsec:ProposedApproach}, the FICE does include them, while the SACE and the AICE do not. 

Moving to the risk-ratio scale,
the FICE, the SACE and the total effect decreased with time from approximately $2.00$ after the first two weeks to $1.20$ (CI95\%: $1.02, 1.39$), $1.31$ (CI95\%: $1.02\%$, $1.59\%$) and 
$1.27$ (CI95\%: $1.07\%$, $1.46\%$),
respectively, during the one-year period.

In summary, while the large-sample bounds are more informative under ORP assumptions then the bounds without ORP assumptions, they do not reject the option of a null result. In contrast, our frailty-based analysis reveals that a causal effect of using Watch instead of Access on future ceftazidime-resistant infection exists. 
The estimated number needed to harm was approximately $65$ (CI95\%: $35$, $472$). This means that for every $10,000$ ios patients treated with the Access antibiotic, had we replaced it with the Watch antibiotic, we would expect about $150$ more patients to acquire a resistant infection. 
The estimated ios proportion was high under both the ORP-type assumptions and the frailty assumptions. This indicates that this subpopulation is highly relevant and that targeting the FICE achieves our original motivation of capturing effects in a large subpopulation.

\section{Discussion}
\label{Sec:Discussion}

In this paper, we proposed a principal stratification approach to study the causal effect on the non-terminal event time in SCR settings.
Our proposed estimand, the FICE$(t)$, is defined in a new stratum among which a well-defined causal effect exists and is informative.
This stratum is more inclusive than previously-proposed subpopulations.
We derive FICE$(t)$ bounds under new ORP assumptions, and provide a principal-stratification oriented examination of the plausibility of these assumptions. As an alternative strategy, we conduct a sensitivity analysis that incorporates a frailty random variable, resembling the approaches of \cite{xu2020bayesian} and \cite{nevo2022causal}.

Our approach offers broad applicability to researchers in a variety of disciplines. One particularly urgent, yet understudied, public‑health problem to which it can be applied to is quantifying the causal effect of antibiotic usage on subsequent antibiotic resistance. To demonstrate its utility, we apply the proposed method to our motivating dataset, estimating the excess ceftazidime-resistance burden from  Watch versus Access use. While in this application the focus was on the non-terminal event, causal effects on the terminal event within strata covered in this paper can also be targeted.

Several avenues remain for future research. First, extending the framework to handle repeated or time-varying antibiotic use.
Second, developing full identification strategies for causal effects on the non-terminal event time under alternative assumptions. Specifically,  assumptions whose violations can be captured by a sensitivity parameter, similarly to the approaches adopted by \cite{ding2017principal} and \cite{zehavi2023matching} for SACE identification.
While our methodology was tailored to address the causal effect of antibiotic usage on future antibiotic resistance,
we believe it may also constitute a step toward strengthening principal stratification tools for studying causal effects in SCR settings more broadly.

\addtocontents{toc}{\protect\setcounter{tocdepth}{3}}

\section{Data Availability Statement}\label{data-availability-statement}

The code for the numerical example and illustrative analysis of the synthetic data is available at  https://github.com/TamirZe/SCR-antibiotics.

\bibliography{references}

\appendix

\setcounter{figure}{0}
\setcounter{table}{0}
\setcounter{equation}{0}
\setcounter{lemma}{0}
\setcounter{proposition}{0}


\renewcommand{\thefigure}{A\arabic{figure}}
\renewcommand{\thetable}{A\arabic{table}}
\renewcommand{\theequation}{A\arabic{equation}}
\renewcommand{\thelemma}{A\arabic{lemma}}
\renewcommand{\theproposition}{A\arabic{proposition}}

\renewcommand{\theequation}{A.\arabic{equation}}
\setcounter{equation}{0}

\spacingset{1.1} 
\section*{Appendix}

\section{Additional information about the motivating dataset}
\label{AppSec:Motivating data}

In this section, we provide information about key variables in the original motivating and matched datasets, and provide the results of the models mentioned in Section \ref{Sec:data}.

We begin with 
Table \ref{Tab:table1}, which presents the distribution of the set of covariates $\bX$ among the Watch group (13,959 patients in the full dataset, 4,143 patients in the matched dataset) and the Access group (5,855 and 4,183 patients), categorized by various forms of demographic and clinical information.
As discussed in Section \ref{Sec:data},
patients in the Watch group were generally older; more likely to suffer from several medical conditions, such as dementia and chronic renal failure;
and more frequently resided in a residential institution prior to admission.
These patterns are expected and align with clinical practice, as the Watch antibiotic is often prescribed to high-risk patients.

\begin{table}[H]
\caption{ \footnotesize Covariates distribution in the full and matched datasets. Matching was on the Mahalanobis distance with a caliper on the estimated propensity score ($c=0.3$ standard deviation of the estimated propensity score).
For discrete covariates, frequencies (proportions) are reported.
For the covariate age,
the mean (standard deviation) is reported.
SMD: standardized mean difference. CRF: chronic renal failure. Cefta: ceftazidime.
} 
\label{Tab:table1} 
\centering
\fbox{
\scriptsize
\begin{tabular}{l|cccccc}
 & \multicolumn{3}{c}{Full dataset} &
   \multicolumn{3}{c}{Matched dataset} \\ 
   & Access & Watch & SMD & Access & Watch & SMD \\
\hline
\textbf{N} & 5,855 & 13,959 &       & 4,143 & 4,143 & \\ 
\textbf{Demographics} &&&&&&\\
 Age & 56.0\,(25.2) & 69.0\,(22.0) & 0.55 & 57.9\,(26.4) & 61.2\,(24.3) & 0.13\\
 Male$^\dagger$ & 3041\,(51.9\%) & 7668\,(54.9\%) & 0.06 & 2192\,(52.9\%) & 2195\,(53\%) & 0.00\\
\hline
\textbf{Sample location} &&&&&&\\
\textbf{at baseline} &&&&&&\\
Other   & 454\,(7.8\%) &  482\,(3.5\%) & 0.19 & 306 (7.4\%) & 262 (6.3\%) & 0.04\\
 Not taken       & 3,093\,(52.8\%) & 7,217\,(51.7\%) & 0.02 & 2481\,(59.9\%) & 2465\,(59.5\%) & 0.01\\
 Multiple sources & 390\,(6.7\%)   & 1,281\,(9.2\%)  & 0.09 & 320\,(7.7\%)   & 302\,(7.3\%)   & 0.02\\
 Urine           & 562\,(9.6\%)   & 3,584\,(25.7\%) & 0.43 & 551\,(13.3\%)  & 580\,(14\%)    & 0.02\\
 Wound           & 1,016\,(17.3\%) & 217\,(1.5\%)   & 0.56 & 196\,(4.7\%)   & 216\,(5.2\%)   & 0.02\\
 Blood           & 182\,(3.1\%)   & 930\,(6.7\%)   & 0.17 & 167\,(4\%)     & 169\,(4.1\%)   & 0.00\\
 Sputum          & 158\,(2.7\%)   & 248\,(1.8\%)   & 0.06 & 122\,(2.9\%)   & 149\,(3.6\%)   & 0.04\\
\hline
\textbf{Arrived from} &&&&&&\\
Other   & 1,074 (18.3\%) &  1,441 (10.3\%) & 0.23 & 762 (18.4\%) & 691 (16.7\%) & 0.04\\
 Home        & 4,098\,(70\%)   & 10,310\,(73.9\%) & 0.09 & 2,784\,(67.2\%) & 2,837\,(68.5\%) & 0.03\\
 Institution & 683\,(11.7\%)  & 2,208\,(15.8\%)  & 0.12 & 597\,(14.4\%)  & 615\,(14.8\%)  & 0.01\\
\hline
\textbf{Hospitalization unit} &&&&&&\\
Other   & 2,445 (41.8\%) &
1,500 (10.7\%) &
0.75 & 1,096 (26.5\%) &
1,018 (24.6\%) & 0.04\\
 Internal & 1,839\,(31.4\%) & 8,883\,(63.6\%) & 0.68 & 1,767\,(42.6\%) & 1,820\,(43.9\%) & 0.03\\
 Surgical & 1,571\,(26.8\%) & 3,576\,(25.6\%) & 0.03 & 1,280\,(30.9\%) & 1,305\,(31.5\%) & 0.01\\
\hline
\textbf{Medical history} &&&&&&\\
 Dementia          & 313\,(5.3\%)  & 1,252\,(9\%)   & 0.14 & 283\,(6.8\%)  & 295\,(7.1\%)  & 0.01\\
 CRF               & 452\,(7.7\%)  & 1,604\,(11.5\%)& 0.13 & 341\,(8.2\%)  & 348\,(8.4\%)  & 0.01\\
 Immunosuppression & 414\,(7.1\%)  & 1,370\,(9.8\%) & 0.10 & 345\,(8.3\%)  & 370\,(8.9\%)  & 0.02\\
 Diabetes          & 1,375\,(23.5\%)& 3,352\,(24\%)  & 0.01 & 856\,(20.7\%) & 866\,(20.9\%) & 0.01\\
 Catheter          & 631\,(10.8\%) & 3,123\,(22.4\%)& 0.32 & 571\,(13.8\%) & 609\,(14.7\%) & 0.03\\
 Previous antibiotic use           & 1,000\,(17.1\%) & 2,049\,(14.7\%) & 0.07 & 699\,(16.9\%) & 729\,(17.6\%) & 0.02\\
 Previous cefta culture ($365$ days)$^\dagger$    & 719\,(12.4\%)  & 1,401\,(10.1\%) & 0.07 & 539\,(13\%)   & 541\,(13.1\%) & 0.00\\
\hline
\textbf{Medical information} &&&&&&\\
 Arrival to culture ($>2$ days)$^\dagger$         & 2,317\,(39.6\%) & 6,016\,(43.1\%) & 0.07 & 1,850\,(44.6\%)& 1,876\,(45.3\%)& 0.01\\
 Arrival to treatment ($>2$ days)$^\dagger$       & 1,041\,(17.8\%) & 3,114\,(22.3\%) & 0.11 & 806\,(19.4\%) & 835\,(20.1\%) & 0.02\\
\hline
\end{tabular}}
\begin{flushleft}
$^\dagger$ \footnotesize
Male is an indicator variable denoting whether the patient is a male. Arrival to culture and Arrival to treatment are indicators for whether arrival preceded culture collection or antibiotic treatment initiation, respectively, by more than two days. Previous ceftazidime culture is an indicator variable denoting whether a ceftazidime culture was taken before baseline.
\end{flushleft}
\end{table}

We now present the results from the models we employed in Section \ref{Sec:data} for data description. 
All models were estimated using the same set of covariates, as in the main analysis.
Table \ref{T2_cox_model} shows the hazard ratio estimates obtained by the Cox regression model for time to death, ignoring infection status, where the observed data were $\tilde{T_2}, \delta_2$, and $\bX$.
Estimated $95\%$ Wald-type confidence intervals are also reported. 
We estimated the propensity score via a logistic regression model. Table \ref{Tab:PS_model} shows the odds ratios estimates along with their estimated $95\%$ Wald-type confidence intervals.
Figure \ref{Fig:Estimated propensity score distribution} depicts the estimated propensity score distributions by antibiotic groups.

\begin{table}[H]
\centering
\caption{Cox model for time-to-death analysis ignoring infection times. HR: hazard ratio; CI95\%: 95\% confidence intervals. Cefta: ceftazidime.}
\label{T2_cox_model}
\begin{tabular}{|l|c|c|}
\hline
\textbf{Covariate} & \textbf{HR} & \textbf{CI95\%} \\
\hline
Access $(A=0)$ & 1.00 &  \\
Watch $(A=1)$ & 0.73 & [0.68, 0.79] \\
\hline
\textbf{Demographics} & & \\
Age & 1.23 & [1.14, 1.33] \\
Age$^{2}$ & 0.998 & [0.997, 0.999] \\
Age$^{3}$ & 1.00 & [1.00 1.00] \\
Male$^\dagger$ & 0.89 & [0.83, 0.95] \\
\hline
\textbf{Sample location at baseline} & & \\
Other & 1.00 & \\
Not taken & 0.68 & [0.57, 0.81] \\
Multiple sources & 0.87 & [0.72, 1.04] \\
Urine & 0.64 & [0.54, 0.77] \\
Wound & 0.50 & [0.38, 0.64] \\
Blood & 0.99 & [0.81, 1.20] \\
Sputum & 0.76 & [0.58, 0.99] \\
\hline
\textbf{Arrived from} & & \\
Other & 1.00 & \\
Home & 0.87 & [0.76, 0.99] \\
Institution & 1.28 & [1.11, 1.48] \\
\hline
\textbf{Hospitalization unit} & & \\
Other & 1.00 & \\
Internal & 1.23 & [1.09, 1.39] \\
Surgical & 0.58 & [0.50, 0.66] \\
\hline
\textbf{Medical history} & & \\
Dementia & 0.89 & [0.81, 0.98] \\
CRF & 1.06 & [0.97, 1.16] \\
Immunosuppression & 1.49 & [1.36, 1.63] \\
Diabetes & 1.04 & [0.97, 1.11] \\
Catheter & 1.79 & [1.67, 1.91] \\
Previous antibiotic (any) & 1.32 & [1.22, 1.43] \\
Previous cefta culture (365 days)$^\dagger$ & 1.35 & [1.22, 1.49] \\
\hline
\textbf{Medical information} & & \\
Arrival to culture ($>$2 days)$^\dagger$ & 0.98 & [0.90, 1.06] \\
Arrival to treatment ($>$2 days)$^\dagger$ & 1.07 & [0.99, 1.16] \\
\hline
\end{tabular}

\begin{flushleft}
$^\dagger$ \footnotesize Male is an indicator variable denoting whether the patient is a male. Arrival to culture and Arrival to treatment are indicators for whether arrival preceded culture collection or antibiotic treatment initiation, respectively, by more than two days. Previous ceftazidime culture is an indicator variable denoting whether a ceftazidime culture was taken before baseline.
\end{flushleft}
\end{table}

\begin{table}[ht]
\centering
\small
\caption{Logistic regression estimates from the propensity score model. OR: odds ratio; CI95\%:  95\% confidence intervals. Cefta: ceftazidime.}
\label{Tab:PS_model}
\begin{tabular}{|l|c|c|}
\hline
\textbf{Covariate} & \textbf{OR} & \textbf{CI95\%} \\
\hline
\textbf{Demographics} & & \\
Age & 0.92 & [0.90, 0.93] \\
Age$^2$ & 1.002
 & [ 1.0017, 1.0023] \\
Age$^3$ & 0.99999 & [0.99998, 0.99999] \\
Male$^\dagger$ & 1.35 & [1.26, 1.46] \\
\hline
\textbf{Baseline sample location} & & \\
Other & 1.00 & \\
Not taken & 1.39 & [1.18, 1.64] \\
Multiple sources & 2.26 & [1.86, 2.75] \\
Urine & 4.88 & [4.09, 5.81] \\
Wound & 0.24 & [0.19, 0.29] \\
Blood & 2.91 & [2.32, 3.65] \\
Sputum & 0.83 & [0.64, 1.09] \\
\hline
\textbf{Arrived from} & & \\
Other & 1.00 & \\
Home & 0.75 & [0.65, 0.85] \\
Institution & 0.49 & [0.42, 0.58] \\
\hline
\textbf{Hospitalization unit} & & \\
Other & 1.00 & \\
Internal & 7.84 & [6.98, 8.80] \\
Surgical & 5.99 & [5.33, 6.75] \\
\hline
\textbf{Medical history} & & \\
Dementia & 0.85 & [0.73, 0.98] \\
CRF & 1.15 & [1.01, 1.31] \\
Immunosuppression & 1.07 & [0.94, 1.22] \\
Diabetes & 0.80 & [0.73, 0.88] \\
Catheter & 1.33 & [1.19, 1.48] \\
Previous antibiotic use & 0.91 & [0.82, 1.01] \\
Previous cefta culture (365 days)$^\dagger$ & 0.78 & [0.69, 0.88] \\
\hline
\textbf{Medical information} & & \\
Arrival to culture ($>$ 2 days)$^\dagger$ & 1.14 & [1.04, 1.25] \\
Arrival to treatment ($>$ 2 days)$^\dagger$ & 1.44 & [1.31, 1.59] \\
\hline
\end{tabular}
\begin{flushleft}
$^\dagger$ \footnotesize Male is an indicator variable denoting whether the patient is a male. Arrival to culture and Arrival to treatment are indicators for whether arrival preceded culture collection or antibiotic treatment initiation, respectively, by more than two days. Previous ceftazidime culture is an indicator variable denoting whether a ceftazidime culture was taken before baseline.
\end{flushleft}
\end{table}

\clearpage

\begin{figure}[H]
\caption{Estimated propensity score distribution by antibiotic-treatment group. }
\label{Fig:Estimated propensity score distribution}
\includegraphics[scale=0.59]{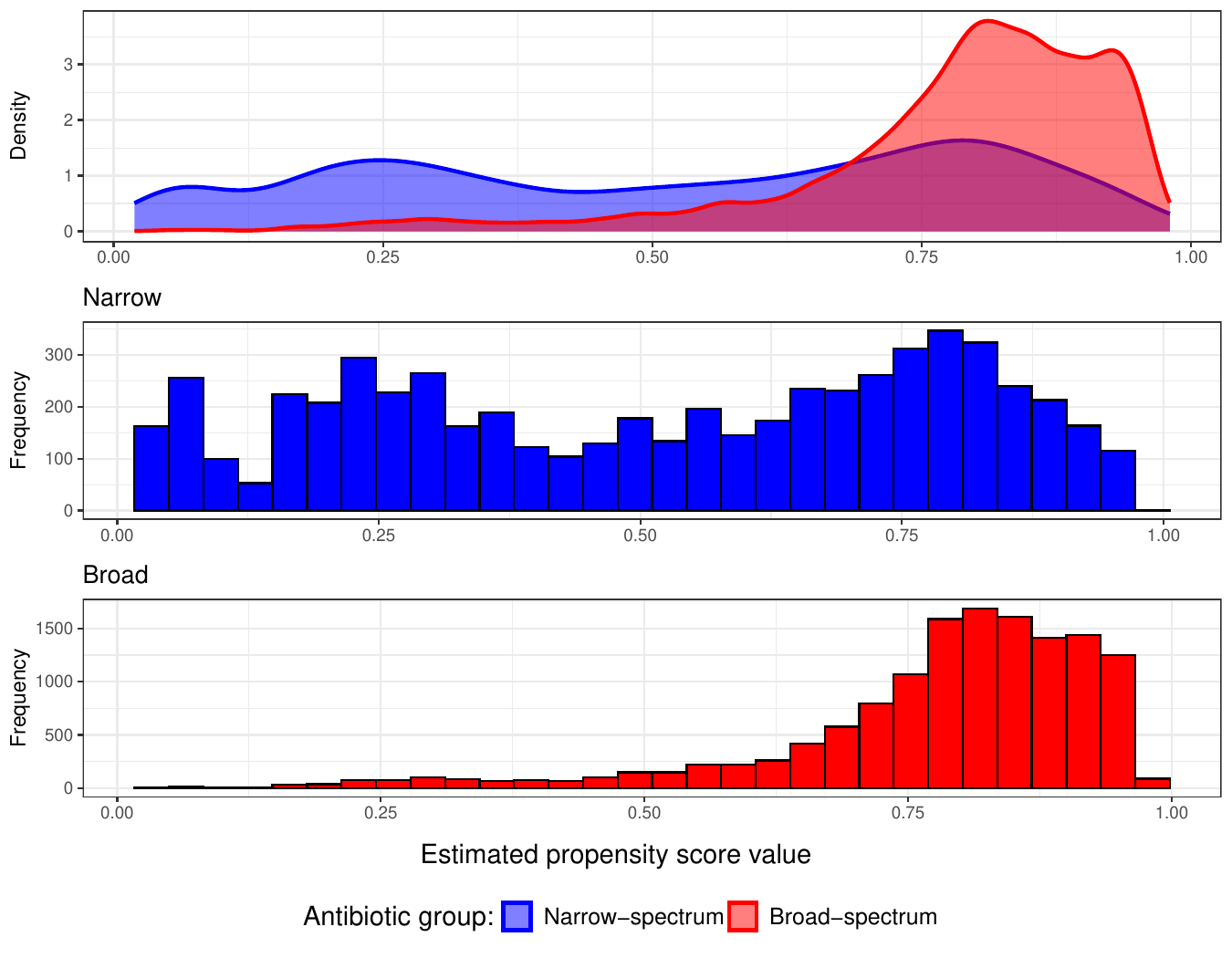}
 \end{figure}
 
\section{Time-varying subpopulations}
\label{AppSec:timevaryingsubpop}

Beyond the stratification based on the one year post-treatment period, we can define different subpopulations as a function of other post-treatment periods $r$, with $r=1$ being the special case analyzed in the main text.
Let
$S_i(a,r)=\mathbbm{1}{ \{T_{i2}(a) > r\} }$, and
$I_i(a,r)=\mathbbm{1}{ \{T_{i1}(a) \le r\} }$ 
be the survival and the resistant infection indicators, respectively, 
within the $r$-length time period
under treatment level $a$.

For $a=0,1$,  let  \textit{ios(r,a)} $= \{i : I_i(r,a)=1 \cup S_i(r,a)=1 \}$.
The subpopulation of patients that would have either acquired an infection or survived within the r-length period under both antibiotic treatments, which we term the \textit{r-infected-or-survivors}, is then defined by 
$ios_r= 
\{i : i \in ios(r,0) \cap ios(r,1)\}$, with stratum proportion $\pi_{ios_r}$.

The proposed estimands can be then redefined to consider these time-varying subpopulations.
For instance, the  TV-FICE$(t,r)$, the effect  
among the $ios_r$ stratum, is defined (on the difference scale) by 
\begin{equation}
\label{TV_FICE_definition}
    \text{TV-FICE}(t,r) = \Pr[T_1(1) \leq t \ |\ ios_r] - \Pr[T_1(0) \leq t | ios_r],    
\end{equation}
for $t\le r$.
The FICE$(t)$, which we focus on in this paper, can be viewed as a special case of the TV-FICE$(t,r)$, taking $r=1$.

The proofs for the bounds can be generalized for each time-period $r$ for which the ORP-type assumptions hold  (see a short discussion in Section \ref{AppSubSec:Bounds}).
The extension for the identification under the frailty assumptions does not require any adjustment.

\section{Patient types illustration}
\label{Appsec:patient types}
Figure \ref{fig:patient_types_illustration}
 provides an example of the possible event times of each patient type described in Table \ref{Tab:stratum_POs_values} of the main text.
For simplicity of presentation only, in the examples below, if event $j$ occurs under the Access antibiotic ($a=0$), then $T_j(0) < T_j(1)$. Further, death times occur only after the last infection time (under both antibiotics). 
Of course, the ordering of the event times might be different as long as the quadruple of potential outcomes defining each patient type is preserved. 
For instance, $T_1(1)$ and/or $T_2(1)$ might be lower than $T_1(0)$.

\begin{figure}[H]
\captionsetup{justification=raggedright, singlelinecheck=false}
\caption{Patient types illustration.}
\label{fig:patient_types_illustration}
\centering
\resizebox{0.8\linewidth}{!}{
\begin{tikzpicture}[
y=1cm, x=2cm, 
thick, font=\normalsize] 
\usetikzlibrary{arrows,decorations.pathreplacing}
\tikzset{
	brace_top/.style={
		decoration={brace},
		decorate
	},
	brace_bottom/.style={
		decoration={brace, mirror},
		decorate
	}
}
\draw[black,line width=1.2pt, ->, >=latex'](0,-1) -- coordinate (x axis) (6,-1) node[right] {}
node[black,left = 12cm] {\large $pt=1$};
\draw (2,-0.9) -- (2,-1.1) node[color = black, below=3pt] {$T_{2i}(0)$};
\draw (4,-0.9) -- (4,-1.1) node[color = black, below=3pt] {$T_{2i}(1)$};
\draw (5,-0.9) -- (5,-1.1) node[color = black, below=3pt] {$t=1$};

\draw[black,line width=1.2pt, ->, >=latex'](0,-2.5) -- coordinate (x axis) (6,-2.5) node[right] {}
node[black,left = 12cm] {\large $pt=2$};
\draw (0.5,-2.4) -- (0.5,-2.6) node[color = black, below=3pt] {$T_{1i}(0)$};
\draw (2,-2.4) -- (2,-2.6) node[color = black, below=3pt] {$T_{2i}(0)$};
\draw (4,-2.4) -- (4,-2.6) node[color = black, below=3pt] {$T_{2i}(1)$};
\draw (5,-2.4) -- (5,-2.6) node[color = black, below=3pt] {$t=1$};

\draw[black,line width=1.2pt, ->, >=latex'](0,-4) -- coordinate (x axis) (6,-4) node[right] {}
node[black,left = 12cm] {\large $pt=3$};
\draw (4,-3.9) -- (4,-4.1) node[color = black, below=3pt] {$T_{2i}(1)$};
\draw (5,-3.9) -- (5,-4.1) node[color = black, below=3pt] {$t=1$};

\draw[black,line width=1.2pt, ->, >=latex'](0,-5.5) -- coordinate (x axis) (6,-5.5) node[right] {}
node[black,left = 12cm] {\large $pt=4$};
\draw (0.5,-5.4) -- (0.5,-5.6) node[color = black, below=3pt] {$T_{1i}(1)$};
\draw (2,-5.4) -- (2,-5.6) node[color = black, below=3pt] {$T_{2i}(0)$};
\draw (4,-5.4) -- (4,-5.6) node[color = black, below=3pt] {$T_{2i}(1)$};
\draw (5,-5.4) -- (5,-5.6) node[color = black, below=3pt] {$t=1$};

\draw[black,line width=1.2pt, ->, >=latex'](0,-7) -- coordinate (x axis) (6,-7) node[right] {}
node[black,left = 12cm] {\large $pt=5$};
\draw (4,-6.9) -- (4,-7.1) node[color = black, below=3pt] {$T_{2i}(0)$};
\draw (5,-6.9) -- (5,-7.1) node[color = black, below=3pt] {$t=1$};

\draw[black,line width=1.2pt, ->, >=latex'](0,-8.5) -- coordinate (x axis) (6,-8.5) node[right] {}
node[black,left = 12cm] {\large $pt=6$};
\draw (0.5,-8.4) -- (0.5,-8.6) node[color = black, below=3pt] {$T_{1i}(0)$};
\draw (2,-8.4) -- (2,-8.6) node[color = black, below=3pt] {$T_{2i}(1)$};
\draw (5,-8.4) -- (5,-8.6) node[color = black, below=3pt] {$t=1$};

\draw[black,line width=1.2pt, ->, >=latex'](0,-10) -- coordinate (x axis) (6,-10) node[right] {}
node[black,left = 12cm] {\large $pt=7$};
\draw (0.5,-9.9) -- (0.5,-10.1) node[color = black, below=3pt] {$T_{1i}(1)$};
\draw (2,-9.9) -- (2,-10.1) node[color = black, below=3pt] {$T_{2i}(0)$};
\draw (5,-9.9) -- (5,-10.1) node[color = black, below=3pt] {$t=1$};

\draw[black,line width=1.2pt, ->, >=latex'](0,-11.5) -- coordinate (x axis) (6,-11.5) node[right] {}
node[black,left = 12cm] {\large $pt=8$};
\draw (0.5,-11.4) -- (0.5,-11.6) node[color = black, below=3pt] {$T_{1i}(0)$};
\draw (2,-11.4) -- (2,-11.6) node[color = black, below=3pt] {$T_{2i}(0)$};
\draw (5,-11.4) -- (5,-11.6) node[color = black, below=3pt] {$t=1$};

\draw[black,line width=1.2pt, ->, >=latex'](0,-13) -- coordinate (x axis) (6,-13) node[right] {}
node[black,left = 12cm] {\large $pt=9$};
\draw (0.5,-12.9) -- (0.5,-13.1) node[color = black, below=3pt] {$T_{1i}(0)$};
\draw (2,-12.9) -- (2,-13.1) node[color = black, below=3pt] {$T_{2i}(0)$};
\draw (5,-12.9) -- (5,-13.1) node[color = black, below=3pt] {$t=1$};

\draw[black,line width=1.2pt, ->, >=latex'](0,-14.5) -- coordinate (x axis) (6,-14.5) node[right] {}
node[black,left = 12cm] {\large $pt=10$};
\draw (0.5,-14.4) -- (0.5,-14.6) node[color = black, below=3pt] {$T_{1i}(0)$};
\draw (1.25,-14.4) -- (1.25,-14.6) node[color = black, below=3pt] {$T_{1i}(1)$};
\draw (2,-14.4) -- (2,-14.6) node[color = black, below=3pt] {$T_{2i}(0)$};
\draw (4,-14.4) -- (4,-14.6) node[color = black, below=3pt] {$T_{2i}(1)$};
\draw (5,-14.4) -- (5,-14.6) node[color = black, below=3pt] {$t=1$};

\draw[black,line width=1.2pt, ->, >=latex'](0,-16) -- coordinate (x axis) (6,-16) node[right] {}
node[black,left = 12cm] {\large $pt=11$};
\draw (5,-15.9) -- (5,-16.1) node[color = black, below=3pt] {$t=1$};

\draw[black,line width=1.2pt, ->, >=latex'](0,-17.5) -- coordinate (x axis) (6,-17.5) node[right] {}
node[black,left = 12cm] {\large $pt=12$};
\draw (0.5,-17.4) -- (0.5,-17.6) node[color = black, below=3pt] {$T_{1i}(0)$};
\draw (1.25,-17.4) -- (1.25,-17.6) node[color = black, below=3pt] {$T_{1i}(1)$};
\draw (4,-17.4) -- (4,-17.6) node[color = black, below=3pt] {$T_{2i}(1)$};
\draw (5,-17.4) -- (5,-17.6) node[color = black, below=3pt] {$t=1$};

\draw[black,line width=1.2pt, ->, >=latex'](0,-19) -- coordinate (x axis) (6,-19) node[right] {}
node[black,left = 12cm] {\large $pt=13$};
\draw (0.5,-18.9) -- (0.5,-19.1) node[color = black, below=3pt] {$T_{1i}(0)$};
\draw (1.25,-18.9) -- (1.25,-19.1) node[color = black, below=3pt] {$T_{1i}(1)$};
\draw (4,-18.9) -- (4,-19.1) node[color = black, below=3pt] {$T_{2i}(0)$};
\draw (5,-18.9) -- (5,-19.1) node[color = black, below=3pt] {$t=1$};

\draw[black,line width=1.2pt, ->, >=latex'](0,-20.5) -- coordinate (x axis) (6,-20.5) node[right] {}
node[black,left = 12cm] {\large $pt=14$};
\draw (0.5,-20.4) -- (0.5,-20.6) node[color = black, below=3pt] {$T_{1i}(0)$};
\draw (5,-20.4) -- (5,-20.6) node[color = black, below=3pt] {$t=1$};

\draw[black,line width=1.2pt, ->, >=latex'](0,-22) -- coordinate (x axis) (6,-22) node[right] {}
node[black,left = 12cm] {\large $pt=15$};
\draw (0.5,-21.9) -- (0.5,-22.1) node[color = black, below=3pt] {$T_{1i}(1)$};
\draw (5,-21.9) -- (5,-22.1) node[color = black, below=3pt] {$t=1$};

\draw[black,line width=1.2pt, ->, >=latex'](0,-23.5) -- coordinate (x axis) (6,-23.5) node[right] {}
node[black,left = 12cm] {\large $pt=16$};
\draw (0.5,-23.4) -- (0.5,-23.6) node[color = black, below=3pt] {$T_{1i}(0)$};
\draw (1.5,-23.4) -- (1.5,-23.6) node[color = black, below=3pt] {$T_{1i}(1)$};
\draw (5,-23.4) -- (5,-23.6) node[color = black, below=3pt] {$t=1$};
\end{tikzpicture}
}
\end{figure}
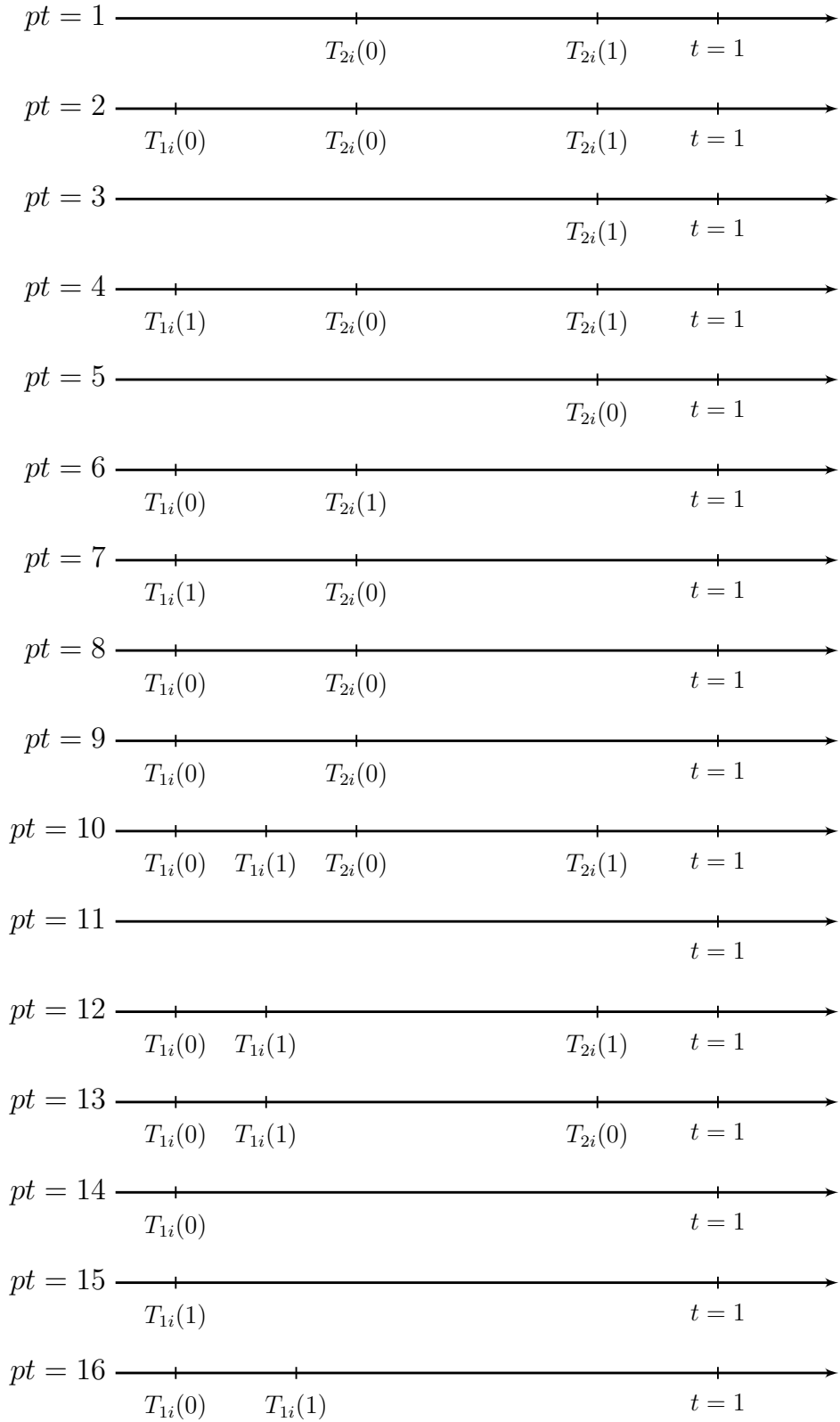

\section{Numerical examples}
\label{AppSec:simulations}
In this section, we present the results of a second simulation scenario we considered, and provide the parameter values and other relevant information about our numerical studies comparing the different estimands in synthetic DGMs.
For both scenarios, the DGM was composed of hazard models, described by the Equations given by \eqref{conditional hazard models DGM}, with one continuous covariate and with
Weibull baseline hazards.
The bivariate frailty variable $\bgamma = (\gamma_0, \gamma_1)$ was Gamma distributed,
with $\theta := \theta_0 = \theta_1$. 
Scenarios were determined by the scale and shape parameters that govern the Weibull distribution of the baseline hazard functions, and the regression coefficients.
The parameter values of each scenario are given in Table \ref{Tab:simulations_parameters_scenarios_all_scenarios}. For each scenario, we calculated the different causal estimands for several $\theta$ and $\rho$ values.

Figure \ref{Fig:initial_simulations_estimands_comparison_scenB} depicts the estimands under the new scenario (Scenario B).
Compared with Scenario A, where the always-infected proportion was negligible (approximately $1\%$), in Scenario B this proportion was remarkably higher. For $\theta=1$, the always-infected proportion ranged from $61\%$ to $63\%$. For $\theta=3$, although lower than under $\theta=1$, it remained relatively high, varying between $35\%$ and $44\%$.
Importantly, the proportion of the ios stratum differed substantially from the proportion of always-survivors.
The former was approximately $78\%$ (under $\theta=1$) and $82\%$ (under $\theta=3$), whereas the latter was lower than in Scenario A, ranging between $23\%$ and $29\%$ (under $\theta=1$) and between $39\%$ and $49\%$ (under $\theta=3$).
Consequently, the FICE and the SACE diverged, in contrast to Scenario A where the ios and the always-survivors strata proportions were similar, and these strata substantially overlapped. 

Interestingly, in Scenario B, the SACE remained zero throughout the entire one-year period, while both the FICE and the AICE were positive; their values increased during the first two months, and then declined. The magnitudes of these estimands diminished as $\rho$ increased, with a more pronounced decrease under the higher $\theta$ value. 
This scenario illustrates that the value and trajectory over time of one estimand can diverge from those of the others, even when all underlying subpopulations are relatively large, and that the FICE is not necessarily monotone in time.

\begin{table}[H]
\footnotesize
\captionsetup{justification=raggedright, singlelinecheck=false}
\caption{Parameters for the data-generating mechanisms under Scenario A (presented in the main text) and Scenario B (presented in the Appendix), each with $\theta = 1,3$ and $\rho = 0,0.5,1$. In every scenario, parameters for the Access group ($A=0$) are shown in the first three rows, and parameters for the Watch group ($A=1$) are shown in the last three rows.}
\label{Tab:simulations_parameters_scenarios_all_scenarios}
\fbox{%
\begin{tabular}{l|llll}
\textbf{Scenario}  
  & \textbf{Baseline hazard shapes} 
  & \textbf{Baseline hazard scales} 
  & \textbf{Covariate coefficients} 
  & \\[0.3em]
\hline\\[-0.8em]
\multirow{6}{*}{A}  & $\tilde{\alpha}^0_{01} = 2.50$ & $\tilde{\mu}^0_{01} = 2.50$ & $\bbeta^0_{01} = (0.00,  -0.69)$ & \\[0.5em]
&  $\tilde{\alpha}^0_{02} = 2.10$ & $\tilde{\mu}^0_{02} = 2.25$ & $\bbeta^0_{02} = (0.00,  0.69)$ &   \\[0.5em]
&  $\tilde{\alpha}^0_{12} = 2.10$ & $\tilde{\mu}^0_{12} = 2.75$ & $\bbeta^0_{12} = (-0.69,  0.69)$ & \\[0.5em]
&  $\tilde{\alpha}^1_{01} = 2.50$ & $\tilde{\mu}^1_{01} = 2.00$ & $\bbeta^1_{01} = (-1.39, 1.10)$ &  \\[0.5em]
&   $\tilde{\alpha}^1_{02} = 2.10$ & $\tilde{\mu}^1_{02} = 2.75$ & $\bbeta^1_{02} = (-0.29,  0.41)$ &   \\[0.5em]
&  $\tilde{\alpha}^1_{12} = 2.10$ & $\tilde{\mu}^1_{12} = 2.25$ & $\bbeta^1_{12} = (0.00,  0.00)$ &  \\[0.5em]
 \hline\\[-0.8em]
\multirow{6}{*}{B} &  $\tilde{\alpha}^0_{01} = 1.00$ & $\tilde{\mu}^0_{01} = 0.10$ & $\bbeta^0_{01} = (0.00,  -0.69)$ & \\[0.5em]
&  $\tilde{\alpha}^0_{02} = 0.50$ & $\tilde{\mu}^0_{02} = 1.00$ & $\bbeta^0_{02} = (0.00,  -0.69)$ &   \\[0.5em]
& $\tilde{\alpha}^0_{12} = 0.50$ & $\tilde{\mu}^0_{12} = 1.00$ & $\bbeta^0_{12} = (0.00,  -0.69)$ & \\[0.5em]
&  $\tilde{\alpha}^1_{01} = 1.00$ & $\tilde{\mu}^1_{01} = 0.10$ & $\bbeta^1_{01} = (0.00,  0.69)$ &  \\[0.5em]
&  $\tilde{\alpha}^1_{02} = 3.00$ & $\tilde{\mu}^1_{02} = 1.00$ & $\bbeta^1_{02} = (0.00,  0.69)$ &   \\[0.5em]
&  $\tilde{\alpha}^1_{12} = 3.00$ & $\tilde{\mu}^1_{12} = 1.00$ & $\bbeta^1_{12} = (0.00,  0.69)$ &  \\[0.5em]
\end{tabular}}
\end{table}

\begin{figure}[H]
\centering
\caption{Comparison of different estimands for a synthetic data-generating mechanism over one year following antibiotic treatment, Scenario B.
}
\label{Fig:initial_simulations_estimands_comparison_scenB}
\includegraphics[scale=0.53] 
{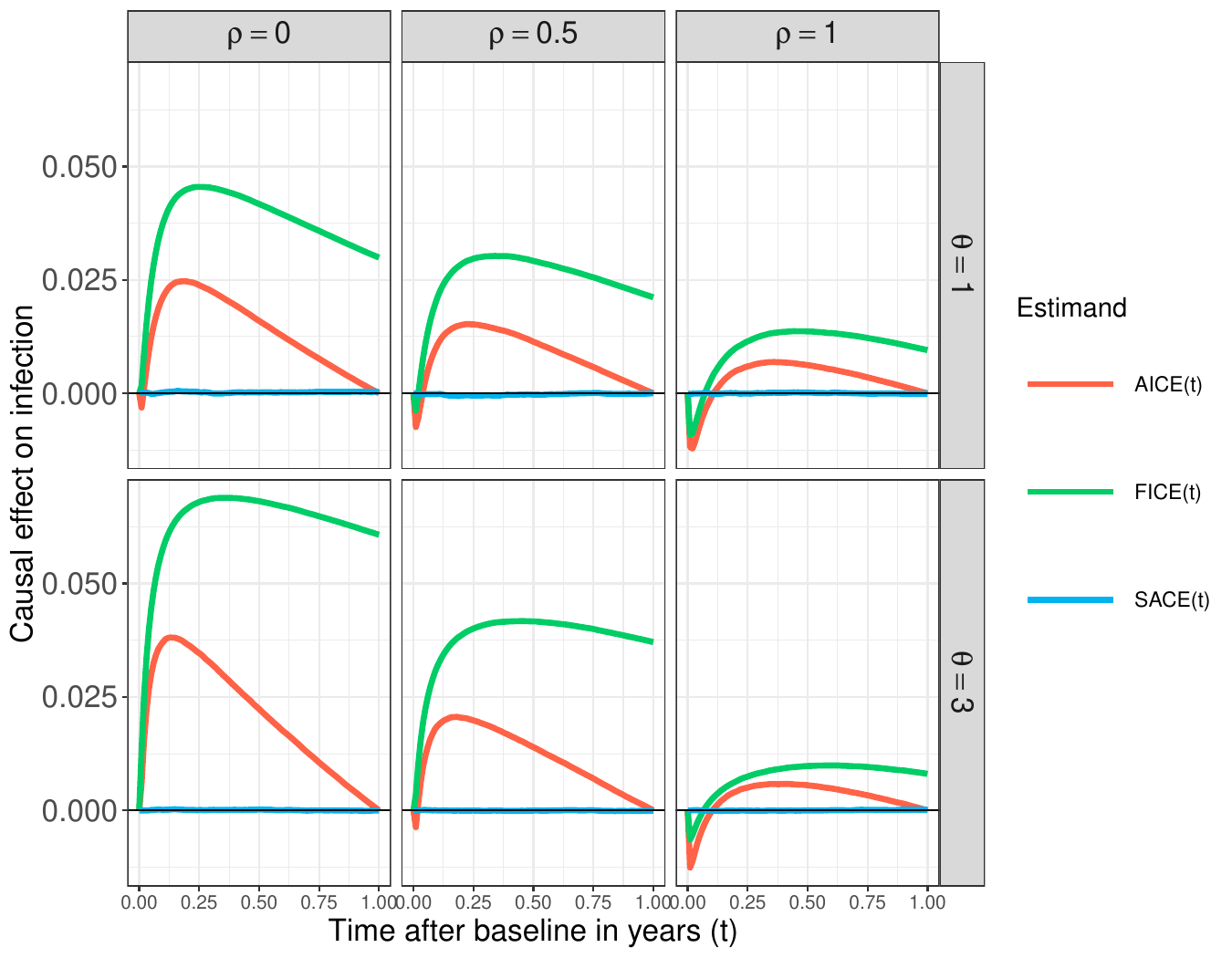}
\end{figure}

\section{FICE identification}
\label{AppSec:Identification}

\subsection{Lemma for FICE identification}
\label{AppSubSec:lemFICE}

The following Lemma will be useful for the construction of the large-sample bounds and for identification under the frailty assumption. We provide the results for general $r$ as defined in Section \ref{AppSec:timevaryingsubpop}. For constructing the bounds given in the main text, we will employ in Section \ref{AppSubSec:Bounds} this lemma with $r=1$. For the frailty-based identification, we will employ the general version of the lemma in Section \ref{Subsec:IdentificationFrailty}.
First, let
\begin{equation}
\label{AppEq:piIOSr}
\pi_{ios_r} =\Pr\Big(\big\{\{T_1(0) \le r\} \cup \{T_2(0) > r\}\big\} \cap \big\{\{T_1(1) \le r\big\} \cup \big\{T_2(1) > r\}\big\}\Big)
\end{equation}
be the proportion of $ios_r$ defined in Section \ref{AppSec:timevaryingsubpop}.
We are now ready to present and prove the lemma.
\begin{lemma}
\label{lem:NumerDenomFICE} For $a=0,1$, and for every $t\le r$,
\begin{equation*}
\Pr\big(T_1(a) \le t\  \vert \  ios_r\big) 
 = \frac{\Pr\Big(\big\{T_1(a) \le t\big\} \cap\big\{\{T_1(1-a) \le r\} \cup \{T_2(1-a) > r\}\big\}\Big)}
{\pi_{ios_r}}
\end{equation*}
\end{lemma}

\begin{proof}
For $a=0,1$ and for $t\le r$,
\begin{align}
\begin{split}
\label{Eq:lemma1}
&\Pr\big(T_1(a) \le t\ |\  ios_r\big)\\[1em]
& = \frac{\Pr\Big(\big\{T_1(a) \le t \big\} \cap \big\{\{T_1(a) \le r\} \cup\{ T_2(a) > r \} \big\} \cap \big\{\{T_1(1-a) \le r\} \cup \{T_2(1-a) > r\}\big\}\Big)}
{\pi_{ios_r}}
\\[1em]
& = \frac{\Pr\Big(\big\{T_1(a) \le t\big\} \cap \big\{\{T_1(1-a) \le r\} \cup \{T_2(1-a) > r\}\big\}\Big)}
{\pi_{ios_r}}
\end{split}
\end{align}
where the second line is by the definition given in Equation \ref{AppEq:piIOSr}, and in the third line we used in the numerator that 
$\{T_1(a) \le t\} \cup \big\{\{T_1(a) \le r\} \cap \{T_2(a) > r\}\big\}=\{T_1(a) \le t\}$ because $r\ge t$ (so $\{T_1(a) \le t\}$ implies $\{T_1(a) \le r\}$. 
\end{proof}

\subsection{Results on ORP Assumptions}
\label{AppSubSec:Assumptions}

The following lemma compares the different estimands and asserts that weak-ORP is indeed the weakest of the three ORP-type assumptions.

\begin{lemma}
\label{lem:orp_assumptions}
Weak-ORP is a weaker assumption than ORP and weaker than ios-ORP.
\end{lemma}

\begin{proof}
First, we show that ORP implies weak-ORP. From ORP, $I(0)=1$ implies $I(1)=1$. By the definition of $ios(1)$, if $I_i(1)=1$, then $i \in ios(1)$. Putting it together, under ORP we have that  for all $i$, $I_i(0)=1$ implies $i\in ios(1)$.
To show that weak-ORP does not imply ORP, we provide a counterexample. Patient type $pt=8$ is not excluded by weak-ORP, but is excluded by ORP. Therefore, if $pt=8$ exists, weak-ORP holds, whereas OPR does not, i.e., weak-ORP does not imply ORP. 

We now show that ios-ORP implies weak-ORP. From the definition of $ios(0)$, if $I_i(0)=1$ then $i \in ios(0)$. By ios-ORP, this means that $i \in ios(1)$.
We obtain that under ios-ORP, $I_i(0)=1$ implies $i\in ios(1)$ for all $i$.
To show that weak-ORP does not imply ios-ORP, we provide another counterexample. Patient type $pt=3$ is not excluded by weak-ORP, but is excluded by ios-ORP. Therefore, if $pt=3$ exists, weak-ORP holds, whereas ios-OPR does not, i.e., weak-ORP does not imply ios-ORP.
\end{proof}






\subsection{Large-sample bounds for the FICE}
\label{AppSubSec:Bounds}

In this section, we provide the proofs for partial identification of FICE$(t)$ with and without ORP-type assumptions. 

We present the proofs for the estimands discussed in the main text, with time-fixed population. The results and proofs are immediately generalized to time-varying subpopulations, by replacing $1$ with $r$, for every $r$ such that  $t \le r\le 1$. Note, however,  that 
the ORP-type assumptions change with $r$, and thus have to hold for the specific $r$ in question for the results to be valid.  
As mentioned in Appendix \ref{AppSec:timevaryingsubpop}, the FICE$(t)$, which we focus on in this paper, can be viewed as a special case of the TV-FICE$(t,r)$ (defined in Equation \ref{TV_FICE_definition}), taking $r=1$.

The following Lemma provides bounds for the numerator of \ref{Eq:lemma1}, and will be used in the proofs of all the bounds we derive.

\begin{lemma}
    \label{lem:BoundNumer}
Under CE and consistency, for $a=0,1$,
\begin{equation*}
\mathcal{L}(a,t) \le \Pr\big(\big\{T_1(a) \le t\big\} \cap \big\{\{T_1(1-a) \le 1\} \cup \{T_2(1-a) > 1\}\big\}\big) \le \mathcal{U}(a,t)
\end{equation*}
where
\begin{align*}
\mathcal{L}(a,t) &= \max\Big\{0, E_{\bX}\big[\Psi_{A=1-a,\bX}\big] - E_{\bX}\big[S_{1|A=a,\bX}(t)\big]\Big\}\\
\mathcal{U}(a,t) &= \min\Big\{E_{\bX}\big[F_{1|A=a,\bX}(t)\big] , E_{\bX}\big[\Psi_{A=1-a,\bX}\big] \Big\}.
\end{align*}

\end{lemma}
\begin{proof}
The upper bound is obtained by
\begin{align*}
\label{numerator_FICE_upper_bound_wout_orp}
\begin{split}
&\Pr\big(\big\{T_1(a) \le t\big\} \cap \big\{\{T_1(1-a) \le 1\} \cup \{T_2(1-a) > 1\}\big\}\big) \\[1ex] 
& \le \min\Big\{\Pr\big(\big\{T_1(a) \le t\big\}\big) , \Pr\big(\big\{\{T_1(1-a) \le 1\} \cup \{T_2(1-a) > 1\}\big\}\big) \Big\}  \\[1ex] 
& = \min\Big\{E_{\bX}\big[F_{1|A=a,\bX}(t)\big] , E_{\bX}\big[\Psi_{A=1-a,\bX}\big] \Big\},
\end{split}
\end{align*}
where the second move is by the law of total expectation, CE, consistency and the definition of $\Psi$.

The lower bound is obtained by
\begin{align*}
&\Pr\big(\big\{T_1(a) \le t\big\} \cap \big\{\{T_1(1-a) \le 1\} \cup \{T_2(1-a) > 1\}\big\}\big) \\[1ex]  
\ge & \max\Big\{0, \Pr\big(\big\{T_1(a) \le t\big\}) + \Pr\big(\big\{\{T_1(1-a) \le 1\} \cup \{T_2(1-a) > 1\}\big\}\big) - 1\Big\}  \\[1ex] 
= & \max\Big\{0, E_{\bX}\big[\Psi_{A=1-a,\bX}\big] - E_{\bX}\Big[S_{1|A=a,\bX}(t)\Big]\Big\},  
\end{align*}
\end{proof}
where the second move is by the law of total expectation, CE, and consistency and the definition of $S_{j|\mathcal{Q}}(t)$.

\subsubsection{Results under ios-ORP (Proposition \ref{Prop:FICE_bounds_ios_ORP})}
\label{AppSubSec:ios_ORP}

Before proving Proposition \ref{Prop:FICE_bounds_ios_ORP}, we present and prove the following lemma that will be useful for the proof. The lemma is of interest in its own right, as it shows that under ios-ORP, the proportion of $ios$ in the population is identifiable. 

\begin{lemma}
\label{lem:PIiosIdentORP}
Under CE, consistency and  ios-ORP, the $ios$ stratum proportion is identifiable by
\begin{equation}
\label{pi_ios_identification_ios_ORP}
\pi_{ios} = E_{\bX}\big[\Psi_{A=0,\bX}\big].
\end{equation} 
\end{lemma}
\begin{proof}
We can write 
\begin{align*}
\label{pi_ios_identification_ios_ORP}
\pi_{ios} & = \Pr\big(\big\{T_1(0) \le 1 \big\} \cup \big\{ T_2(0) > 1\big\}\big) \\[0.25em]
& = E_{\bX}\big[\Pr\big(\big\{T_1(0) \le 1 \big\} \cup \big\{ T_2(0) > 1\big\})|\bX\big]\\[0.25em]
&= E_{\bX}\big[\Pr\big(\big\{T_1(0) \le 1 \big\} \cup \big\{ T_2(0) > 1\big\}\big)|A=0, \bX\big]\\[0.25em]
& = E_{\bX}\big[\Psi_{A=0,\bX}\big],
\end{align*}
where the first move is by ios-ORP, the second is by the law of total expectation, the third by CE and the final line is by consistency and the definition of $\Psi$.
\end{proof}

We are now ready for the proof of Proposition \ref{Prop:FICE_bounds_ios_ORP}. 
\paragraph{Proof of Proposition \ref{Prop:FICE_bounds_ios_ORP}}

\begin{proof}
\label{proof_prop_FICE_ios_ORP}

Starting from $\Pr(T_1(0) \le t| ios)$, the denominator in Lemma  \ref{lem:NumerDenomFICE},  is identified by Lemma \ref{lem:PIiosIdentORP}. Turning to the numerator, for every $t \le 1$, $T_{i1}(0) \le t$ implies that $i \in ios(0)$. Therefore, by ios-ORP, 
either  $T_{i1}(1) \le 1$ or $T_{i2}(1) > 1$. Hence, 
\begin{equation*}
\Pr\big(\big\{T_1(0) \le t\big\} \cap \big\{\{T_1(1) \le 1\} \cup \{T_2(1) > 1\}\big\}\big) =  \Pr\big(\big\{T_1(0) \le t\big\}\big) = E_{\bX}\big[F_{1|A=0,\bX}(t)\big],
\end{equation*}
where the second move is by the law of total expectation, CE, and consistency.
Therefore, $\Pr(T_1(0) \le t| ios)$ is identified by
\begin{equation}
\label{Eq:T10iosORP}
\Pr(T_1(0) \le t| ios) = \dfrac{ E_{\bX}\big[F_{1|A=0,\bX}(t)\big]}
{ E_{\bX}\big[\Psi_{A=0,\bX}\big]}.
\end{equation}

Regarding $\Pr(T_1(1) \le t| ios)$, we again employ Lemma  \ref{lem:NumerDenomFICE};   the denominator is identified by Lemma \ref{lem:PIiosIdentORP}. 
To derive upper and lower bounds for the numerator, 
we apply Lemma \ref{lem:BoundNumer} and  obtain that
\begin{equation}
\label{Eq:T11iosORP}
\max\Bigg\{0, 1 - 
\frac{E_{\bX}\big[S_{1|A=1, \bX}(t)\big]} {E_{\bX}\big[\Psi_{A=0,\bX}\big]}\Bigg\}  
 \le \Pr(T_1(1) \le t| ios) \le \min\Bigg\{1, \frac{E_{\bX}\big[F_{1|A=1,\bX}(t)\big]} {E_{\bX}\big[\Psi_{A=0,\bX}\big]}\Bigg\}.
\end{equation}
The proof is done by subtracting \eqref{Eq:T10iosORP} from \eqref{Eq:T11iosORP}.
\end{proof}

\subsubsection{Results under weak-ORP (Proposition \ref{Prop:FICE_bounds_weak_ORP})}
\label{AppSubSec:weak_ios_ORP}

\paragraph{Bounds for the infected-or-survivors stratum proportion}\mbox{}\\
\label{Appsubsubsec:Bounds for the ios proportion}

Before providing FICE bounds, 
we first derive bounds for the ios stratum proportion under weak-ORP.

\begin{lemma}
\label{lem:PIios_weakORP_bounds}
Under CE, consistency and weak-ORP, the $ios$ stratum proportion is bounded by
\begin{equation*}
\mathcal{\tilde{L}}_{\pi} \le \pi_{ios} \le \mathcal{\tilde{U}}_{\pi},
\end{equation*}
where
\begin{align*}
\mathcal{\tilde{L}}_{\pi} &= \max\Big\{E_{\bX}\big[F_{1|A=0,\bX}(1)\big], 
E_{\bX}\big[\Psi_{A=0,\bX}\big]  + E_{\bX}\big[\Psi_{A=1,\bX}\big] - 1\Big\},\\
\mathcal{\tilde{U}}_{\pi} &= \min\Big\{E_{\bX}\big[\Psi_{A=0,\bX}\big],
E_{\bX}\big[\Psi_{A=1,\bX}\big] + 
E_{\bX}\big[\Pr\big(\big\{T_1 \le 1\big\} \cap \big\{T_2 \le 1\big\}|A=0,\bX\big)\big]
\Big\}.
\end{align*}
\end{lemma}

\begin{proof}
First, we note that under weak-ORP, $\pi_{ios}$ is equal to the difference between the $ios(0)$ subpopulation proportion
and the proportion of $pt=3$ (which is excluded by ios-ORP, but not by weak-ORP, and belongs to $ios(0)$, but not to the $ios$ stratum).
Therefore, an upper bound under weak-ORP can be derived by
\begin{align}
\label{pi_ios_upper_bound_weak_orp}
\begin{split}
\pi_{ios}  
& = \Pr\big(\big\{T_1(0) \le 1 \big\} \cup \big\{ T_2(0) > 1\big\}\big) \\[0.3em]
&- \Pr\big(\big\{\{T_1(0) > 1\} \cap \{T_2(0) > 1\}\big\} \cap
\big\{\{T_1(1) > 1\} \cap \{T_2(1) \le 1\}\big\}\big) \\[0.3em]
& = \Pr\big(\big\{T_1(0) \le 1 \big\} \cup \big\{ T_2(0) > 1\big\}\big) \\[0.3em]
&- \Pr\big(\big\{T_2(0) > 1\big\} \cap \big\{\{T_1(1) > 1\} \cap \{T_2(1) \le 1\}\big\}\big) \\[0.3em]
& \le \Pr\big(\big\{T_1(0) \le 1 \big\} \cup \big\{ T_2(0) > 1\big\}\big)  \\[0.3em]
& - \max\Big\{0,\Pr\big(\big\{T_2(0) > 1\big\}\big) + \Pr\big(\big\{T_1(1) > 1\big\} \cap \big\{T_2(1) \le 1\big\}\big) - 1\Big\}\\[0.3em]
& = \Pr\big(\big\{T_1(0) \le 1 \big\} \cup \big\{ T_2(0) > 1\big\}\big)  \\[0.3em]
& - \max\Big\{0,\Pr\big(\big\{T_2(0) > 1\big\}\big) - \Pr\big(\big\{T_1(1) \le 1\big\} \cup \big\{T_2(1) > 1\big\}\big) \Big\} \\[0.3em]
& = \min\Big\{\Pr\big(\big\{T_1(0) \le 1 \big\} \cup \big\{ T_2(0) > 1\big\}\big), \\[0.3em]
&\Pr\big(\big\{T_1(1) \le 1 \big\} \cup \big\{ T_2(1) > 1\big\}\big) + \Pr\big(\big\{T_1(0) \le 1\big\} \cap \big\{T_2(0) \le 1\big\}\big)\Big\} \\[0.3em]
& = \min\Big\{E_{\bX}\big[\Psi_{A=0,\bX}\big],
E_{\bX}\big[\Psi_{A=1,\bX}\big] + 
E_{\bX}\big[\Pr\big(\big\{T_1(0) \le 1\big\} \cap \big\{T_2(0) \le 1\big\}|A=0,\bX\big)\big]
\Big\},
\end{split}
\end{align}
where the second move is by by weak-ORP (with contrapositive);
The fourth is because 
$\big\{T_1(1) \le 1\big\} \cup \big\{T_2(1) > 1\big\}$ is the complementary of
$\big\{T_1(1) > 1\big\} \cap \big\{T_2(1) \le 1\big\}$;
The fifth is because by $\Pr(A \cup B) = \Pr(A) + \Pr(B) - \Pr(A \cap B)$ (for two events $A$ and $B$), along with the law of total probability, it follows that

\begin{align*}
& \Pr\big(\big\{T_1(0) \le 1 \big\} \cup \big\{ T_2(0) > 1\big\}\big) - \Pr\big(\big\{ T_2(0) > 1\big\}\big) \\ 
=& \Pr\big(\big\{T_1(0) \le 1\big\}\big) - \Pr\big(\big\{T_1(0) \le 1\big\} \cap \big\{T_2(0) > 1\big\}\big)
=\Pr\big(\big\{T_1(0) \le 1\big\} \cap \big\{T_2(0) \le 1\big\}\big) . 
\end{align*}

Now, a lower bound under weak-ORP can be constructed by
\begin{align}
\label{pi_ios_lower_bound_weak_orp}
\begin{split}
\pi_{ios}  
& = \Pr\big(\big\{T_1(0) \le 1 \big\} \cup \big\{ T_2(0) > 1\big\}\big) \\[0.3em]
&- \Pr\big(\big\{\{T_1(0) > 1\} \cap \{T_2(0) > 1\}\big\} \cap
\big\{\{T_1(1) > 1\} \cap \{T_2(1) \le 1\}\big\}\big) \\[0.3em]
& \ge \Pr\big(\big\{T_1(0) \le 1\big\} \cup \big\{ T_2(0) > 1\big\}\big) \\[0.3em] 
&- \min\Big\{\Pr\big(\big\{T_1(0) > 1\big\} \cap \big\{T_2(0) > 1\big\}\big), \: \Pr\big(\big\{T_1(1) > 1\big\} \cap \big\{T_2(1) \le 1\big\}\big)\Big\}\\[0.3em]
& = \max\Big\{\Pr\big(\big\{T_1(0) \le 1\big\}\big), \\[0.3em]
& \Pr\big(\big\{T_1(0) \le 1 \big\} \cup \big\{ T_2(0) > 1\big\}\big) - \Pr\big(\big\{T_1(1) > 1\big\} \cap \big\{T_2(1) \le 1\big\}\big)\Big\} \\[0.3em] 
& = \max\Big\{E_{\bX}\big[F_{1|A=0,\bX}(1)\big], 
E_{\bX}\big[\Psi_{A=0,\bX}\big]  + E_{\bX}\big[\Psi_{A=1,\bX}\big] - 1\Big\},
\end{split}
\end{align}

where the third move is because by $\Pr(A \cap B) = \Pr(A) + \Pr(B) - \Pr(A \cup B)$ (for two events $A$ and $B$), along with the law of total probability, it follows that
\begin{align*}
&\Pr\big(\big\{T_1(0) \le 1\big\} \cup \big\{T_2(0) > 1\big\}\big) \\[0.3em]
-&\Pr\big(\big\{T_1(0) > 1\big\} \cap \big\{T_2(0) > 1\big\}\big) = \\[0.3em]
&\Pr\big(\big\{T_1(0) \le 1\big\}\big) + \Pr\big(\big\{T_2(0) > 1\big\}\big) -\Pr\big(\big\{T_1(0) \le 1\big\} \cap \big\{T_2(0) > 1\big\}\big) \\[0.3em]
-&\Pr\big(\big\{T_1(0) > 1\big\} \cap \big\{T_2(0) > 1\big\}\big) \\[0.3em]  =& \Pr\big(\big\{T_1(0) \le 1\big\}\big), \nonumber
\end{align*}
and the last move is due to the law of total expectation, CE, consistency and the definition of $\Psi$ (and its complementary event).
\end{proof}

Of note is that the upper bound under weak-ORP is lower than (or equal to) the expression identifying the $ios$ stratum proportion under ios-ORP.

We are now ready for the proof of Proposition \ref{Prop:FICE_bounds_weak_ORP}. 
\paragraph{Proof of Proposition \ref{Prop:FICE_bounds_weak_ORP}}
\label{Proof of Proposition FICE_bounds_weak_ORP}

\begin{proof}
\label{proof_prop_FICE_weak_ORP}

We first provide upper and lower bounds, denoted by $\tilde{u}(t)$ and $\tilde{l}(t)$ respectively, for the difference between the numerators of Lemma \ref{lem:NumerDenomFICE}, i.e. for
\begin{equation}
\label{numer_diff_lemma1}
\Pr\big(\big\{T_1(1) \le t\big\} \cap \big\{\{T_1(0) \le 1\} \cup \big\{T_2(0) > 1\big\}\big\}\big) - \Pr\big(\big\{T_1(0) \le t\big\} \cap \big\{\{T_1(1) \le 1\} \cup \{T_2(1) > 1\}\big\}\big).
\end{equation}
Statrting with the first term ($a=1$), we employ Lemma  \ref{lem:BoundNumer} to derive lower and upper bounds, i.e.,
$\mathcal{L}(1,t) \le \Pr\big(\big\{T_1(1) \le t\big\} \cap \big\{\{T_1(0) \le 1\} \cup \{T_2(0) > 1\}\big\}\big) \le \mathcal{U}(1,t)$.

For the second term ($a=0$), we note that for every $t \le 1$, $T_{i1}(0) \le t$ implies that $T_{i1}(0) \le 1$. Therefore, by weak-ORP,  either  $T_{i1}(1) \le 1$ or $T_{i2}(1) > 1$. Hence, 
\begin{equation}
\label{numerA0_lemma1}
\Pr\big(\big\{T_1(0) \le t\big\} \cap \big\{\{T_1(1) \le 1\} \cup \{T_2(1) > 1\}\big\}\big) =  \Pr\big(\big\{T_1(0) \le t\big\}\big) = E_{\bX}\Big[F_{1|A=0,\bX}(t)\Big],
\end{equation}
where the second move is by the law of total expectation, CE, and consistency.
Of note is that the proofs for each of the numerators of Lemma \ref{lem:NumerDenomFICE} under weak-ORP are identical to the proofs of Proposition \ref{Prop:FICE_bounds_ios_ORP} (i.e., under ios-ORP).

The upper and lower bounds for the difference in 
\ref{numer_diff_lemma1}, $\tilde{{u}}(t)$ and ${\tilde{l}}(t)$, are obtained by substracting \eqref{numerA0_lemma1} from $\mathcal{U}(1,t)$ and from $\mathcal{L}(1,t)$, respectively,
\begin{align}
\begin{split}
\label{bounds_diff_numer_weak_ORP}
& \tilde{{u}}(t) = \min\Big\{E_{\bX}\big[F_{1|A=1,\bX}(t)\big], E_{\bX}\big[\Psi_{A=0,\bX}(1)\big]\Big\} - E_{\bX}\big[F_{1|A=0,\bX}(t)\big],  \\
& \tilde{{l}}(t) = \max\Big\{0, E_{\bX}\big[\Psi_{A=0,\bX}(1)\big] - E_{\bX}\big[S_{1|A=1,\bX}(t)\big]\Big\} - E_{\bX}\big[F_{1|A=0,\bX}(t)\big], 
\end{split}
\end{align}
with $1\{\cdot\}$ being the indicator function. 

To obtain FICE bound, we tie \eqref{bounds_diff_numer_weak_ORP} together with Lemma \ref{lem:NumerDenomFICE} and the bounds for $\pi_{ios}$ under weak-ORP, given in Lemma \ref{lem:PIios_weakORP_bounds}.
For the upper FICE bound, if $\tilde{{u}}(t)$ is positive, we divide it by $\mathcal{\tilde{{L}}}_{\pi}$. Otherwise, we divide it by $\mathcal{{\tilde{U}}}_{\pi}$.
For the lower FICE bound, if $\tilde{{l}}(t)$ is positive, we divide it by $\mathcal{{\tilde{U}}}_{\pi}$. Otherwise, we divide it by $\mathcal{{\tilde{L}}}_{\pi}$.
\end{proof}

\subsubsection{Results without ORP-type assumptions (Proposition \ref{Prop:FICE_bounds_wout_ORP})}
\label{AppSubSec:wout_ORP}

Before providing FICE bounds, 
we first derive bounds for the ios stratum proportion without ORP-type assumptions.

\begin{lemma}
\label{lem:PIios_wout_ORP_bounds}
Under CE and consistency, the $ios$ stratum proportion is bounded by
\begin{equation*}
\mathcal{\dot{L}_{\pi}} \le \pi_{ios} \le \mathcal{\dot{U}_{\pi}},
\end{equation*}
where
\begin{align*}
\mathcal{\dot{L}_{\pi}} &= \max\Big\{0, E_{\bX}\big[\Psi_{A=0,\bX}\big] + E_{\bX}\big[\Psi_{A=1,\bX}\big] - 1\Big\},\\
\mathcal{\dot{U}_{\pi}} &= \min_a\Big\{E_{\bX}\big[\Psi_{A=a,\bX}\big] \Big\}.
\end{align*}
\end{lemma}

\begin{proof}
First,
\begin{align}
\begin{split}
\pi_{ios} 
\le & \min_a\Big\{\Pr\big(\big\{\{T_1(a) \le 1\} \cup \{T_2(a) > 1\}\big\} \big)\Big\}  \\
= & \min_a\Big\{E_{\bX}\big[\Psi_{A=a,\bX}\big] \Big\}, 
\nonumber
\end{split}
\end{align}
where the second move is by the law of total expectation, CE, and consistency.

Second,
\begin{align}
\begin{split}
\pi_{ios} 
\ge & \max\Big\{0, \Pr\big(\big\{\{T_1(0) \le 1\} \cup \{T_2(0) > 1\}\big\} \big) + \Pr\big(\big\{\{T_1(1) \le 1\} \cup \{T_2(1) > 1\}\big\} \big) - 1\Big\}  \\
= & \max\Big\{0, E_{\bX}\big[\Psi_{A=0,\bX}\big] + E_{\bX}\big[\Psi_{A=1,\bX}\big] - 1\Big\},
\end{split}
\nonumber
\end{align}
where the last move is due to the law of total expectation, CE, consistency and definition of $\Psi$.
\end{proof}

Similarly to the upper bound under weak-ORP (Lemma \ref{lem:PIios_weakORP_bounds}), the upper bound without ORP assumptions is lower than (or equal to) the expression identifying the $ios$ stratum proportion under ios-ORP.

Furthermore, the upper bound without ORP assumptions is lower than (or equal to) its counterpart under weak-ORP.
When $E_{\bX}\big[\Psi_{A=0,\bX}\big] \le E_{\bX}\big[\Psi_{A=1,\bX}\big]$, 
the upper bounds coincide and eqaul to the proportion under ios-ORP, that is, $E_{\bX}\big[\Psi_{A=0,\bX}\big]$.
This is the case in our motivating dataset.
The lower bounds without ORP assumptions and under weak-ORP were also identical in the motivating dataset.

We now turn to provide FICE bounds without ORP-type assumptions.

\begin{proposition}
\label{Prop:FICE_bounds_wout_ORP}
Under consistency and CE, the FICE$(t)$ is bounded by
\begin{equation*}
\mathcal{\dot{L}}(t) \le \: \:   \text{FICE$(t)$} \le \:  \mathcal{\dot{U}}(t),  
\end{equation*}
where    $\mathcal{\dot{U}}(t) = \dfrac{\dot{u}(t)}{1\{\dot{u}(t) \ge 0\} \cdot \mathcal{\dot{L}}_{\pi} + 1\{\dot{u}(t) < 0\} \cdot \mathcal{\dot{U}}_{\pi}} \:, $
$\mathcal{\dot{L}}(t) = \dfrac{\dot{l}(t)}{1\{\dot{l}(t) \ge 0\} \cdot \mathcal{\dot{U}}_{\pi} + 1\{\dot{l}(t) < 0\} \cdot \mathcal{\dot{L}}_{\pi}}$,
and
\begin{align*}
& \dot{u}(t) = \min\Big\{E_{\bX}\big[F_{1|A=1,\bX}(t)\big], E_{\bX}\big[\Psi_{A=0,\bX}\big]\Big\} - \max\Big\{0, E_{\bX}\big[\Psi_{A=1,\bX}\big] - E_{\bX}\big[S_{1|A=0,\bX}(t)\big]\Big\},  \\
& \dot{l}(t) = \max\Big\{0, E_{\bX}\big[\Psi_{A=0,\bX}\big] -  E_{\bX}\big[S_{1|A=1,\bX}(t)\big]\Big\} - \min\Big\{E_{\bX}\big[F_{1|A=0,\bX}(t)\big], E_{\bX}\big[\Psi_{A=1,\bX}\big]\Big\}. 
\end{align*}
\end{proposition}

\noindent Note that both $\dot{u}(t)$ and $\dot{l}(t)$ are higher than (or equal to) their counterparts under weak-ORP, $\tilde{{u}}(t)$ and $\tilde{{l}}(t)$.

\begin{proof}
\label{proof_prop_FICE_wout_ORP}
We first derive upper and lower bounds for the difference between the numerators of Lemma \ref{lem:NumerDenomFICE}, i.e. for
\begin{equation}
\label{numer_diff_lemma1_second}
\Pr\big(\big\{T_1(1) \le t\big\} \cap \big\{\big\{T_1(0) \le 1\big\} \cup \big\{T_2(0) > 1\big\}\big\}\big) - \Pr\big(\big\{T_1(0) \le t\big\} \cap \big\{\big\{T_1(1) \le 1\big\} \cup \big\{T_2(1) > 1\big\}\big\}\big) .
\end{equation}

Employing Lemma \ref{lem:BoundNumer} with both $a=0$ and $a=1$, the upper bound for
\eqref{numer_diff_lemma1_second} is obtained by subtracting the lower bound of the second term ($a=0$) from the upper bound of the first term ($a=1$),
\begin{align}
\label{upper_bound_diff_numer_wout_ORP}
& \dot{u}(t) = \min\Big\{E_{\bX}\big[F_{1|A=1,\bX}(t)\big], E_{\bX}\big[\Psi_{A=0,\bX}\big]\Big\} - \max\Big\{0, E_{\bX}\big[\Psi_{A=1,\bX}\big] - E_{\bX}\big[S_{1|A=0,\bX}(t)\big]\Big\}.
\end{align}
Similarly, the lower bound for
\eqref{numer_diff_lemma1_second} is obtained by subtracting the upper bound of the second term ($a=0)$ from the lower bound of the first term ($a=1$),
\begin{align}
\label{lower_bound_diff_numer_wout_ORP}
& \dot{l}(t) = \max\Big\{0, E_{\bX}\big[\Psi_{A=0,\bX}\big] -  E_{\bX}\big[S_{1|A=1,\bX}(t)\big]\Big\} - \min\Big\{E_{\bX}\big[F_{1|A=0,\bX}(t)\big], E_{\bX}\big[\Psi_{A=1,\bX}\big]\Big\}.
\end{align}

To obtain FICE bound, we tie \eqref{upper_bound_diff_numer_wout_ORP} and \eqref{lower_bound_diff_numer_wout_ORP} together with Lemma \ref{lem:NumerDenomFICE} and the bounds for $\pi_{ios}$ of Lemma \ref{lem:PIios_wout_ORP_bounds}.
For the upper FICE bound, if $\dot{u}(t)$ is positive, we divide it by $\mathcal{\dot{L}_{\pi}}$. Otherwise, we divide it by $\mathcal{\dot{U}_{\pi}}$.
For the lower FICE bound, if $\dot{l}(t)$ is positive, we divide it by $\mathcal{\dot{U}}_{\pi}$. Otherwise, we divide it by $\mathcal{\dot{L}_{\pi}}$. 
\end{proof}

\subsubsection{Bounds for the risk-ratio scale FICE$(t)$ (Propositions \ref{Prop:FICE_bounds_ratio_weak_ORP} and \ref{Prop:FICE_bounds_ratio_wout_ORP})}
\label{Appsubsec:bounds for the risk-ratio scale estimand}

\begin{proposition}
\label{Prop:FICE_bounds_ratio_weak_ORP}
Under consistency, CE and weak-ORP (or ios-ORP), the risk-ratio scale FICE$(t)$ is bounded by
\begin{align}
\nonumber
\Bigg[\dfrac{\max\Big\{0, E_{\bX}\big[\Psi_{A=0,\bX}\big] - E_{\bX}\big[S_{1|A=1,\bX}(t)\big]\Big\}}{E_{\bX}\big[F_{1|A=0,\bX}(t)\big]} ; \:
 \dfrac{\min\Big\{E_{\bX}\big[\Psi_{A=0,\bX}\big], E_{\bX}\big[F_{1|A=1,\bX}(t)\big]\Big\}}{E_{\bX}\big[F_{1|A=0,\bX}(t)\big]}\Bigg].
\end{align}
\end{proposition}

\begin{proof}
Applying Lemma \ref{lem:NumerDenomFICE} for both $\Pr\big(T_1(1) \le t\  \vert \  ios\big)$ and $\Pr\big(T_1(0) \le t\  \vert \  ios\big)$, it follows that the ratio $\dfrac{\Pr\big(\big\{T_1(1) \le t\big\}\  \vert \  ios\big)}{\Pr\big(\big\{T_1(0) \le t\big\}\  \vert \  ios\big)}$ can be written as  
\begin{equation}
\label{Lemma1_ratio}
\dfrac{\Pr\big(\big\{T_1(1) \le t\big\}\  \vert \  ios\big)}{\Pr\big(\big\{T_1(0) \le t\big\}\  \vert \  ios\big)} 
 = \dfrac{\Pr\big(\big\{T_1(1) \le t\big\} \cap \big\{\big\{T_1(0) \le 1\big\} \cup \big\{T_2(0) > 1\big\}\big\}\big)}{\Pr\big(\big\{T_1(0) \le t\big\} \cap \big\{\big\{T_1(1) \le 1\big\} \cup \big\{T_2(1) > 1\big\}\big\}\big)}.
\end{equation}
To derive upper and lower bounds for the numerator of \eqref{Lemma1_ratio}, we  employ Lemma \ref{lem:BoundNumer} (with $a=1$) to obtain

\begin{equation}
\label{Eq:T11iosORP_risk_ratio}
\mathcal{L}(1,t) \le \Pr\big(\big\{T_1(1) \le t\big\} \cap \big\{\{T_1(0) \le 1\} \cup \{T_2(0) > 1\}\big\}\big) \le \mathcal{U}(1,t),
\end{equation}
where
\begin{align*}
\mathcal{L}(1,t) &= \max\Big\{0, E_{\bX}\big[\Psi_{A=0,\bX}\big] - E_{\bX}\big[S_{1|A=1,\bX}(t)\big]\Big\},\\
\mathcal{U}(1,t) &= \min\Big\{E_{\bX}\big[\Psi_{A=0,\bX}\big], E_{\bX}\big[F_{1|A=1,\bX}(t)\big]  \Big\}.
\end{align*}

The denominator of \eqref{Lemma1_ratio}, i.e. for $a=0$, is identified by
\begin{equation}
\label{numerA0_lemma1_risk_ratio}
\Pr\big(\big\{T_1(0) \le t\big\} \cap \big\{\{T_1(1) \le 1\} \cup \{T_2(1) > 1\}\big\}\big) =  \Pr\big(\big\{T_1(0) \le t\big\}\big) = E_{\bX}\big[F_{1|A=0,\bX}(t)\big],
\end{equation}
where the first move is by weak-ORP (and because for every $t \le 1$, $T_{i1}(0) \le t$ implies that $T_{i1}(0) \le 1$), and the second is by the law of total expectation, CE, and consistency.
The proof is done by dividing \eqref{Eq:T11iosORP_risk_ratio} by \eqref{numerA0_lemma1_risk_ratio}.

We note that on the risk-ratio scale, identification of $\pi_{ios}$ is not required, as it cancels out. Hence, the bounds under ios-ORP and weak-ORP coincide, because only the (full/partial) identification of $\pi_{ios}$ distinguishes between them. 
\end{proof}

\begin{proposition}
\label{Prop:FICE_bounds_ratio_wout_ORP}
Under consistency and CE, the risk-ratio scale FICE$(t)$ is bounded by
\begin{align}
 \Bigg[\dfrac{\max\Big\{0, E_{\bX}\big[\Psi_{A=0,\bX}\big] - E_{\bX}\big[S_{1|A=1,\bX}(t)\big]\Big\}}{\min\Big\{E_{\bX}\big[\Psi_{A=1,\bX}\big], E_{\bX}\big[F_{1|A=0,\bX}(t)\big]\Big\}} ; \:
 \dfrac{\min\Big\{E_{\bX}\big[F_{1|A=1,\bX}(t)\big], E_{\bX}\big[\Psi_{A=0,\bX}\big]\Big\}}{\max\Big\{0, E_{\bX}\big[\Psi_{A=1,\bX}\big] - E_{\bX}\big[S_{1|A=0,\bX}(t)\big]\Big\}}\Bigg].
\end{align}
\end{proposition}

\begin{proof}
Applying Lemma \ref{lem:NumerDenomFICE} for both $\Pr\big(T_1(1) \le t\  \vert \  ios\big)$ and $\Pr\big(T_1(0) \le t\  \vert \  ios\big)$, it follows that the ratio $\dfrac{\Pr\big(\big\{T_1(1) \le t\big\}\  \vert \  ios\big)}{\Pr\big(\big\{T_1(0) \le t\big\}\  \vert \  ios\big)}$ can be written as  
\begin{equation}
\label{Lemma1_ratiO_wout_ORP}
\dfrac{\Pr\big(\big\{T_1(1) \le t\big\}\  \vert \  ios\big)}{\Pr\big(\big\{T_1(0) \le t\big\}\  \vert \  ios\big)} 
 = \dfrac{\Pr\big(\big\{T_1(1) \le t\big\} \cap \big\{\big\{T_1(0) \le 1\big\} \cup \big\{T_2(0) > 1\big\}\big\}\big)}{\Pr\big(\big\{T_1(0) \le t\big\} \cap \big\{\big\{T_1(1) \le 1\big\} \cup \big\{T_2(1) > 1\big\}\big\}\big)}.
\end{equation}

Similarly to the proof of Proposition \ref{Prop:FICE_bounds_ratio_weak_ORP}, to derive upper and lower bounds for the numerator of \eqref{Lemma1_ratiO_wout_ORP}, we employ Lemma \ref{lem:BoundNumer} with $a=1$,
\begin{equation}
\label{Eq:T11iosORP_risk_ratio_wout}
\mathcal{L}(1,t) \le \Pr\big(\big\{T_1(1) \le t\big\} \cap \big\{\{T_1(0) \le 1\} \cup \{T_2(0) > 1\}\big\}\big) \le \mathcal{U}(1,t),
\end{equation}
where
\begin{align*}
\mathcal{L}(1,t) &= \max\Big\{0, E_{\bX}\big[\Psi_{A=0,\bX}\big] - E_{\bX}\big[S_{1|A=1,\bX}(t)\big]\Big\}\\
\mathcal{U}(1,t) &= \min\Big\{E_{\bX}\big[\Psi_{A=0,\bX}\big], E_{\bX}\big[F_{1|A=1,\bX}(t)\big]  \Big\}.
\end{align*}
Now, for upper and lower bounds for the denominator of \eqref{Lemma1_ratiO_wout_ORP}, we employ Lemma \ref{lem:BoundNumer} with $a=0$, 
\begin{equation}
\label{Eq:T10iosORP_risk_ratio_wout}
\mathcal{L}(0,t) \le \Pr\Big(\big\{T_1(0) \le t\big\} \cap \big\{\{T_1(1) \le 1\} \cup \{T_2(1) > 1\}\big\}\Big) \le \mathcal{U}(0,t),
\end{equation}
where
\begin{align*}
\mathcal{L}(0,t) &= \max\Big\{0, E_{\bX}\big[\Psi_{A=1,\bX}\big] - E_{\bX}\big[S_{1|A=0,\bX}(t)\big]\Big\}\\
\mathcal{U}(0,t) &= \min\Big\{E_{\bX}\big[\Psi_{A=1,\bX}\big], E_{\bX}\big[F_{1|A=0,\bX}(t)\big]  \Big\}.
\end{align*}

To finish the proof, we tie \eqref{Lemma1_ratiO_wout_ORP} together with \eqref{Eq:T11iosORP_risk_ratio_wout} and \eqref{Eq:T10iosORP_risk_ratio_wout}.
The upper FICE bound is obtained by $\dfrac{\mathcal{U}(1,t)}{\mathcal{L}(0,t)}$.
The lower FICE bound is obtained by  $\dfrac{\mathcal{L}(1,t)}{\mathcal{U}(0,t)}$.
\end{proof}

\subsection{Identification under the frailty assumptions}
\label{Subsec:IdentificationFrailty}

\subsubsection{Results  for the TV-FICE$(t,r)$ under the frailly assumptions (Proposition \ref{Prop:FICE_identification_frailty})}
\label{AppSubsec:Proof_FICE_frailty}

We now present the proof for the identification formula of the TV-FICE$(t,r)$ under the frailty assumptions. 
Here, we consider the time-varying population, determined by the $r$ value.
As noted in Section \ref{AppSubSec:Bounds}, the FICE$(t)$
can be viewed as a special case of the TV-FICE$(t,r)$, taking $r=1$.
In contrast to the case under the ORP-type assumptions, frailty assumptions require no modification, i.e., the frailty assumptions are sufficient for   identification of the TV-FICE$(t,r)$ as a function of $\rho$ in all the subpopulations.

\begin{proof}
\label{proof:FICE_identification_frailty}

By the law of total expectation and by the definition of $\pi_{ios_r}$, Equation \eqref{Eq:lemma1} of Lemma \ref{lem:NumerDenomFICE} can be written as
\footnotesize
\begin{align} 
\label{Eq:LIE_over_gamma}
\dfrac{E_{\bX,\bgamma}\Big[\Pr\Big(\big\{T_1(a) \le t\big\} \cap \big\{\big\{T_1(1-a) \le r\big\} \cup \big\{T_2(1-a) > r\big\}\big\} \Bigl\lvert \bX,\bgamma\Big)\Big]}{E_{\bX,\bgamma}\Big[\Pr\Big(\big\{\big\{T_1(0) \le r\big\} \cup \big\{T_2(0) > r\big\}\big\} \cap \big\{\big\{T_1(1) \le r\big\} \cup \big\{T_2(1) > r\big\}\big\} \Bigl\lvert \bX,\bgamma\Big)\Big]}.
\end{align}
\normalsize
Now, Equation \eqref{Eq:LIE_over_gamma} can be replaced with
\footnotesize
\begin{align} 
& \dfrac{E_{\bX,\bgamma}\Big[\Pr\Big(\big\{T_1 \le t \big\}\Bigl\lvert A=a, \bX,\gamma_a\Big) \Pr\Big(\big\{\big\{T_1 \le r\big\} \cup \big\{T_2 > r\big\} \big\}\Bigl\lvert A=1-a, \bX,\gamma_{1-a}\Big)\Big]}{E_{\bX,\bgamma}\Big[\Pr\Big(\big\{\big\{T_1 \le r\big\} \cup \big\{T_2 > r\big\} \big\} \Bigl\lvert A=0, \bX,\gamma_0\Big) \Pr\Big(\big\{\big\{T_1 \le r\big\} \cup \big\{T_2 > r\big\} \big\} \Bigl\lvert A=1, \bX, \gamma_1 \Big)\Big]} \nonumber \\[0.4em]
= & \dfrac{\int_{\bx} \int_{\gamma_0} \int_{\gamma_1} \Pr\Big(\big\{T_1 \le t \big\}\Bigl\lvert A=a, \bX=\bx,\gamma_a\Big) \Pr\Big(\big\{\big\{T_1 \le r\big\} \cup \big\{T_2 > r\big\}\big\} \Bigl\lvert A=1-a, \bX=\bx,\gamma_{1-a}\Big) f(\bgamma, \bx) d\bgamma_1 d\bgamma_0 d\bx}{\int_{\bx} \bigg\{\int_{\gamma_0} \int_{\gamma_1} \Pr\Big(\big\{\big\{T_1 \le r\big\} \cup \big\{T_2 > r\big\}\big\} \Bigl\lvert A=0, \bX=\bx,\gamma_0\Big) \Pr\Big(\big\{\big\{T_1 \le r\big\} \cup \big\{T_2 > r\big\}\big\} \Bigl\lvert A=1, \bX=\bx, \gamma_1 \Big)}
\nonumber \\
& \cdot f(\bgamma, \bx) d\bgamma_1 d\bgamma_0 d\bx\bigg\}
\nonumber \\[0.4em]
= & \dfrac{\int_{\bx} f(\bx) \int_{\gamma_0} \int_{\gamma_1} \Pr\Big(\big\{T_1 \le t\big\} \Bigl\lvert A=a, \bX=\bx,\gamma_a\Big) \Pr\Big(\big\{\big\{T_1 \le r\big\} \cup \big\{T_2 > r\big\}\big\} \Bigl\lvert A=1-a, \bX=\bx,\gamma_{1-a}\Big) f_{\btheta}(\bgamma) d\bgamma_1 d\bgamma_0 d\bx}{\int_{\bx}\bigg\{ f(\bx)\int_{\gamma_0} \int_{\gamma_1} \Pr\Big(\big\{\big\{T_1 \le r\big\} \cup \big\{T_2 > r\big\}\big\} \Bigl\lvert A=0, \bX=\bx,\gamma_0\Big) \Pr\Big(\big\{\big\{T_1 \le r\big\} \cup \big\{T_2 > r\big\}\big\} \Bigl\lvert A=1, \bX=\bx, \gamma_1 \Big)} \nonumber \\ 
& \cdot f_{\btheta}(\bgamma) d\bgamma_1 d\bgamma_0 d\bx\bigg\} \nonumber \\
= & \dfrac{E_{\bX}\Big[E_{\bgamma}\Big[\Pr\Big(\big\{T_1 \le t \big\}\Bigl\lvert A=a, \bX,\gamma_a\big) \Pr\Big(\big\{\big\{T_1 \le r\big\} \cup \big\{T_2 > r\big\}\big\} \Bigl\lvert A=1-a, \bX,\gamma_{1-a}\Big)\Big]\Big]}{E_{\bX}\Big[E_{\bgamma}\Big[\Pr\Big(\big\{\big\{T_1 \le r\big\} \cup \big\{T_2 > r\big\}\big\} \Bigl\lvert A=0, \bX,\gamma_0\Big) \Pr\Big(\big\{\big\{T_1 \le r\big\} \cup \big\{T_2 > r\big\} \big\} \Bigl\lvert A=1, \bX, \gamma_1 \Big)\Big]\Big]}, \label{identification frailty probability under a value}
\end{align}
\normalsize
where the first term is due to parts (i) and (ii) of the frailty assumptions, together with consistency, the third is by
part (iii) of the frailty assumptions, and where $f(\bgamma, \bx)$ is the joint density function of $\bgamma$ and $\bX$, $f(\bx)$ is the density function of $\bX$ and $f_{\btheta}(\bgamma)$ is the density function of $\bgamma$.
\end{proof}

\subsubsection{Results for the risk-ratio scale TV-FICE$(t,r)$ under the frailty assumptions (Proposition \ref{Prop:FICE_bounds_ratio_frailty})}

\begin{proposition}
\label{Prop:FICE_bounds_ratio_frailty}
Under the frailty assumptions, the risk-ratio scale TV-FICE$(t,r)$ is identified by
\begin{align} \dfrac{E_{\bX}\Big[E_{\bgamma}\Big[\Pr\big(\big\{T_1 \le t \big\}\Bigl\lvert A=1, \bX,\gamma_1\big) \Pr\big(\big\{\big\{T_1 \le r\big\} \cup \big\{T_2 > r\big\}\big\} \Bigl\lvert A=0, \bX,\gamma_{0}\big)\Big]\Big]}{E_{\bX}\Big[E_{\bgamma}\Big[\Pr\big(\big\{T_1 \le t \big\}\Bigl\lvert A=0, \bX,\gamma_0\big) \Pr\big(\big\{\big\{T_1 \le r\big\} \cup \big\{T_2 > r\big\}\big\} \Bigl\lvert A=1, \bX,\gamma_{1}\big)\Big]\Big]}. \nonumber
\end{align}
\end{proposition}

\begin{proof}
The risk-ratio scale estimand is obtained by taking the ratio between the term in (\ref{identification frailty probability under a value}) with $a=1$ and the same term with $a=0$. The denominator cancels out, so we divide the numerators of both terms to obtain the risk-ratio scale TV-FICE$(t,r)$ bounds.
\end{proof}

\clearpage
\section{Estimation}
\label{Appsubsec:Estimation}
\subsection{Derivation of the likelihood function}
\label{Appsubsubsec:likelihood}

For completeness, we derive here the likelihood function for the illness-death models. In practice it is maximized by an EM algorithm.
First, recall the notations adopted for the parameters in Section \ref{Sec:estimation};
$\widetilde{\btheta} = (\theta_0, \theta_1),
\widetilde{\bbeta} = (\bbeta^0, \bbeta^1),   
\widetilde{\blambda}_0(t) = (\blambda^0(t), \blambda^1(t))$, and
$\widetilde{\bLambda}_0(t) =(\bLambda^0(t),\bLambda^1(t))$, where, for $a=0,1$,
$\bbeta^a = (\bbeta^a_{01}, \bbeta^a_{02}, \bbeta^a_{12}), \\ 
\blambda^a(t) = (\lambda^0_{01}(t|a), \lambda^0_{02}(t|a), \lambda^0_{12}(t|a))$, and
$\bLambda^a(t) = (\Lambda^0_{01}(t|a), \Lambda^0_{02}(t|a), \Lambda^0_{12}(t|a))$. 

Under the hazard models 
described in Equation \eqref{conditional hazard models DGM}, the likelihood function, integrated over the
frailty distribution, equals 
$L(\widetilde{\btheta}, \widetilde{\bLambda}_0, \widetilde{\blambda}_0, \widetilde{\bbeta}) = \Pi_{i=1}^{n}L_i$, where
\small
\begin{align}
L_i &= \int_{0}^{\infty} f(\widetilde{T}_{i1}, \widetilde{T}_{i2}, \delta_{i1}, \delta_{i2}, \gamma_{A_i} | \bX_i) d\bgamma_{A_i}  \nonumber \\[1em]
& = \int_{0}^{\infty} f(\widetilde{T}_{i1}, \widetilde{T}_{i2}, \delta_{i1}, \delta_{i2} | \bX_i, \gamma_{A_i} ) f(\gamma_{A_i} | \bX_i) d\gamma_{A_i} = \int_{0}^{\infty}f(\widetilde{T}_{i1}, \widetilde{T}_{i2}, \delta_{i1}, \delta_{i2}, | \bX_i, \gamma_{A_i} ) f(\gamma_{A_i}) d\gamma_{A_i} \nonumber \\[1em]
&= \int_{0}^{\infty}  \Big[ \lambda_{01}(\widetilde{T}_{i1}|A_i, \bX_i, \gamma_{A_i}) \Big]^{\delta_{i1}} 
\Big[ \lambda_{02}(\widetilde{T}_{i1}|A_i, \bX_i, \gamma_{A_i}) \Big]^{(1 - \delta_{i1}) \delta_{i2}} 
\Big[ \lambda_{12}(\widetilde{T}_{i2}|A_i, \bX_i, \gamma_{A_i})  \Big]^{\delta_{i1} \delta_{i2}} \nonumber \\[1em]
& \exp\bigg\{-\big[H_{01}(\widetilde{T}_{i1}|A_i, \bX_i, \gamma_{A_i}) + H_{02}(\widetilde{T}_{i1}|a_i, \bX_i, \gamma_{A_i}) + \delta_1\big(H_{12}(\widetilde{T}_{i2}|a_i, \bX_i, \gamma_{A_i}) - H_{12}(\widetilde{T}_{i1}|A_i, \bX_i, \gamma_{A_i})\big)\big] \bigg\}
\nonumber \\[1em]
& f(\gamma_{A_i}) d\gamma_{A_i} \nonumber \\[1em]
&= \int_{0}^{\infty} \Big[ \gamma_{A_i} \lambda^0_{01}(\widetilde{T}_{i1}|A_i) \exp(\bX_i^t\bbeta^{A_i}_{01}) \Big]^{\delta_{i1}} 
\Big[ \gamma_{A_i} \lambda^0_{02}(\widetilde{T}_{i2}|A_i) \exp(\bX_i^t\bbeta^{A_i}_{02}) \Big]^{(1 - \delta_{i1}) \delta_{i2}} 
\Big[ \gamma_{A_i} \lambda^0_{12}(\widetilde{T}_{i2}|A_i) \exp(\bX_i^t\bbeta^{A_i}_{12})  \Big]^{\delta_{i1} \delta_{i2}} \nonumber \\[1em]
& \exp\bigg\{-\gamma_{a_i}\bigg[\Big(H^0_{01}(\widetilde{T}_{i1}|A_i)\exp(\bX_i^t\bbeta^{A_i}_{01}) + H^0_{02}(\widetilde{T}_{i1}|A_i)\exp(\bX_i^t\bbeta^{A_i}_{02})\Big) + \nonumber \\[1em] 
& \Big(\big(H^0_{12}(\widetilde{T}_{i2}|A_i) - H^0_{12}(\widetilde{T}_{i1}|A_i)\big)  \exp(\bX_i^t\bbeta^{A_i}_{12}) \delta_1\Big) \bigg] \bigg\}  f(\gamma_{A_i}) d\gamma_{A_i} \nonumber \\[1em]
&= \lambda'_i \int_{0}^{\infty} \gamma_{A_i}^{\delta_i'} 
\exp(-\gamma_{A_i}k_i)
\frac{ \gamma_{A_i}^{1 / \theta_{a_i} - 1} \exp(-(1 / \theta_{A_i})\gamma_{A_i}) }{ \Gamma(1 / \theta_{A_i}) \theta_{A_i}^{1 / \theta_{A_i}} }  d\gamma_{A_i},
\label{plugging lambda' in the likelihhod}
\end{align}
\normalsize

where
\footnotesize
\begin{align}
\lambda'_i &= \exp\Big( \delta_{i1}X_i^t\bbeta^{A_i}_{01} +
(1 - \delta_{i1})\delta_{i2}X_i^t \bbeta^{A_i}_{02} + 
\delta_{i1} \delta_{i2}X_i^t \bbeta^{A_i}_{12}\Big)
\Big[ \lambda^0_{01}(\widetilde{T}_{i1}|A_i) \Big]^{\delta_{i1}} 
\Big[ \lambda^0_{02}(\widetilde{T}_{i2}|A_i) \Big]^{(1 - \delta_{i1}) \delta_{i2}} 
\Big[ \lambda^0_{12}(\widetilde{T}_{i2}|A_i)  \Big]^{\delta_{i1} \delta_{i2}}, \nonumber \\
\delta'_i &= \delta_{i1} + \delta_{i2}, \: \: \text{and}
\nonumber \\
k_i &= H^0_{01}(\widetilde{T}_{i1}|A_i)\exp(\bX_i^t\bbeta^{A_i}_{01}) + H^0_{02}(\widetilde{T}_{i1}|A_i)\exp(\bX_i^t\bbeta^{A_i}_{02}) + \big[H^0_{12}(\widetilde{T}_{i2}|A_i) - H^0_{12}(\widetilde{T}_{i1}|A_i)\big]  \exp(\bX_i^t\bbeta^{A_i}_{12})\delta_{i1}.
\nonumber
\end{align}
\normalsize

Define $\alpha_i' = \frac{1}{\theta} + \delta_i'$.
Now, the term in \eqref{plugging lambda' in the likelihhod} can be written as 
\small
\begin{align}
& \lambda'_i  \frac{\Gamma(\alpha_i') \theta_{A_i}^{\alpha_i'}}{\Gamma(1 / \theta_{A_i}) \theta_{A_i}^{1/\theta_{A_i}}} \int_{0}^{\infty} \frac{ \gamma_{A_i}^{\alpha_i' - 1} \exp(-(1 / \theta_{A_i})\gamma_{A_i}) }{ \Gamma(\alpha_i') \theta_{A_i}^{\alpha_i'} } 
\exp(-k_i\gamma_{A_i})
d\gamma_{A_i} \nonumber \\[1em]
& = \lambda'_i \frac{\Gamma(\alpha_i')} {\Gamma(1 / \theta_{A_i})} \cdot \theta_{A_i}^{\delta_i'} E_{G_i}\big[
\exp(-k_i G_i)\big]\nonumber \\[1em]
& = \lambda'_i \frac{\Gamma(\alpha_i')} {\Gamma(1 / \theta_{A_i})} \cdot \theta_{A_i}^{\delta_i'} LP_{Gi}(k_i),
\label{adding Laplace transform notaion}
\end{align}
where for every $i$, $G_i$ is a Gamma random variable, with a shape parameter $\alpha_i'$ and a scale parameter $\theta_{A_i}$, and  $LP_{\omega}(v) = E_{\omega}\big[
\exp(-v \omega)\big]$ is the Laplace transform of a Gamma random variable $\omega$, at the point $v$. 
Now, plugging the Laplace transform of a Gamma random variable, \eqref{adding Laplace transform notaion} can be replaced by
\small
\begin{align}
& = \lambda'_i  \frac{\Gamma(\alpha_i')} {\Gamma(1 / \theta_{A_i})} \cdot \theta_{A_i}^{\delta_i'}\Big[\frac{1/\theta_{A_i}}{k_i + 1/\theta_{A_i}}\Big]^{\alpha_i'} \nonumber \\[1em]
& = \lambda'_i  \frac{\Gamma(\alpha_i')} {\Gamma(1 / \theta_{A_i})} \cdot 
\Big[\frac{\theta_{A_i}}{1 + \theta_{A_i}k_i}\Big]^{\delta_i'}
\Big[1 + \theta_{A_i}k_i \Big]^{-1/\theta_{A_i}}
\nonumber \\
& = \lambda'_i \: \frac{\Gamma(1 / \theta_{A_i} + \delta'_i)} {\Gamma(1 / \theta_{A_i})} \cdot \Big[\frac{\theta_{a_i}}{1 + \theta_{A_i}k_i}\Big]^{\delta_i'} 
\phi_{A_i}^{(0)}(k_i), \label{likelihood with LP before derivatives} 
\end{align}
where for $a=0,1, \: q=0,1,2$, $\phi_{a}^{(q)}(k)$ is the $q$-th derivative of 
$LP_{\gamma_{a}}(k)$,
with respect to $k$.

\noindent Using the fact that $\Gamma(u + 1)= u\Gamma(u)$, and that $\delta'_i \in \{0,1,2\}$ is a non-negative integer,  we obtain
\small
\begin{align}
&  L_i = \lambda_i' \: (-1)^
{\delta_i'} \phi_{A_i}^{(\delta_i')}(k_i).
\nonumber
\end{align}
\normalsize
\noindent Therefore, 
the likelihood function, marginalized over the frailty distribution,
equals to
\begin{align*}
\begin{split}
L(\widetilde{\btheta}, \widetilde{\bLambda}_0, \widetilde{\blambda}_0, \widetilde{\bbeta}) & = \Pi_{i=1}^{n} \Big\{\ \Big[ \lambda^0_{01}(\widetilde{T}_{i1}|A_i) \Big]^{\delta_{i1}} 
\Big[ \lambda^0_{02}(\widetilde{T}_{i1}|A_i) \Big]^{(1 - \delta_{i1}) \delta_{i2}} 
\Big[ \lambda^0_{12}(\widetilde{T}_{i1}|A_i)  \Big]^{\delta_{i1} \delta_{i2}} \\
& \exp\Big( \delta_{i1}X_i^t\bbeta^{A_i}_{01} +
(1 - \delta_{i1})\delta_{i2}X_i^t \bbeta^{A_i}_{02} + 
\delta_{i1} \delta_{i2}X_i^t \bbeta^{A_i}_{12}\Big) (-1)^
{\delta_{i1} + \delta_{i2}} \phi_{a_i}^{(\delta_{i1} + \delta_{i2})}(k_i) \Big\}.
\end{split}
\end{align*}

\subsection{Expectation maximization algorithm}
\label{Appsubsec:EM_algorithm}

We implemented the EM algorithm proposed by \cite{nevo2022causal} for estimation.
First, we set the initial values of $\widetilde{\bbeta}$ and $\widetilde{\bLambda}_0(t)$ according to the estimates obtained from standard Cox regression models, with $\gamma_{A_i}=1$ for each patient $i=1,...,n$. The three steps outlined below were then repeated until convergence.

\begin{enumerate}[label=(\alph*)]
    \item E-step:
    \begin{enumerate}[label=(\roman*)]
    \item 
    As shown by \cite{nevo2022causal}, for each patient $i=1,...,n$, the posterior expectation of $\gamma_{a_i}$ given the observed data can be written as
    \begin{align*}
    E[\gamma_{A_i}|D_i] = \dfrac{(-1)^{\delta_i^{'} + 1} \phi_{A_i}^{\delta_i^{'} + 1}(k_i)}{(-1)^{\delta_i^{'}} \phi_{A_i}^{\delta_i^{'}}(k_i)},
\end{align*}
where $D_i$ is the observed data for patient $i$.
For each patient $i$ we then estimated $\widehat{E}[\log\gamma_{A_i}|D_i]$ and $\widehat{E}[\gamma_{A_i}|D_i]$ using the current values of $\widetilde{\bbeta}$, $\widetilde{\bLambda}_0(t)$ (which are both required to obtain $k_i$) and $\widetilde{\btheta}$ (which is required for $\phi_{0}$ and $\phi_{1}$).
    \end{enumerate}
    \item  M-step:
    \begin{enumerate}[label=(\roman*)]
    \item We estimated $\widetilde{\bbeta}$ with six Cox regression models, with offset terms $\log\widehat{E}[\gamma_{A_i}|D_i]$.
    We estimated $\widetilde{\bLambda}_0(t)$ using Breslow estimators.
    \item
    To estimate $\widetilde{\btheta}$,
    we estimated $\theta_a$ (for $a=0,1$) by maximizing the conditional expectation of the estimated log-likelihood of $\gamma_a$ within treatment group $A=a$, given the observed data and the current parameter values.

    \end{enumerate}
\end{enumerate}

 The algorithm adopted to the specific case where the frailty bivariate random variable is Gamma distributed is given in \cite{nevo2022causal}.

\section{Application to the motivating example - additional information and figures}
\label{Appsec: Application}

\subsection{Bounds}
\label{Appsubsec:bounds}
Figure \ref{Fig:res_upper_bound_under_weak_ORP_ratio} below shows the upper bound for the FICE$(t)$ on the risk-ratio scale under weak-ORP. The risk-ratio scale estimand declined over time, from approximately $2.00$ two weeks after baseline to $1.27$ (CI95\%: $1.07$, $1.46$) after one year.
 \begin{figure}[H]
\centering
\caption{\footnotesize{Estimated upper bounds for the FICE$(t)$ on the risk-ratio scale under weak-ORP and ios-ORP. The bounds under weak-ORP and under ios-ORP coincide analytically.
The lower bounds under weak-ORP and without ORP-type assumptions were zero, i.e, non-informative.
The upper bound without ORP-type assumptions was also non-informative.}}
\label{Fig:res_upper_bound_under_weak_ORP_ratio}
\centering
\includegraphics[scale=0.5] 
{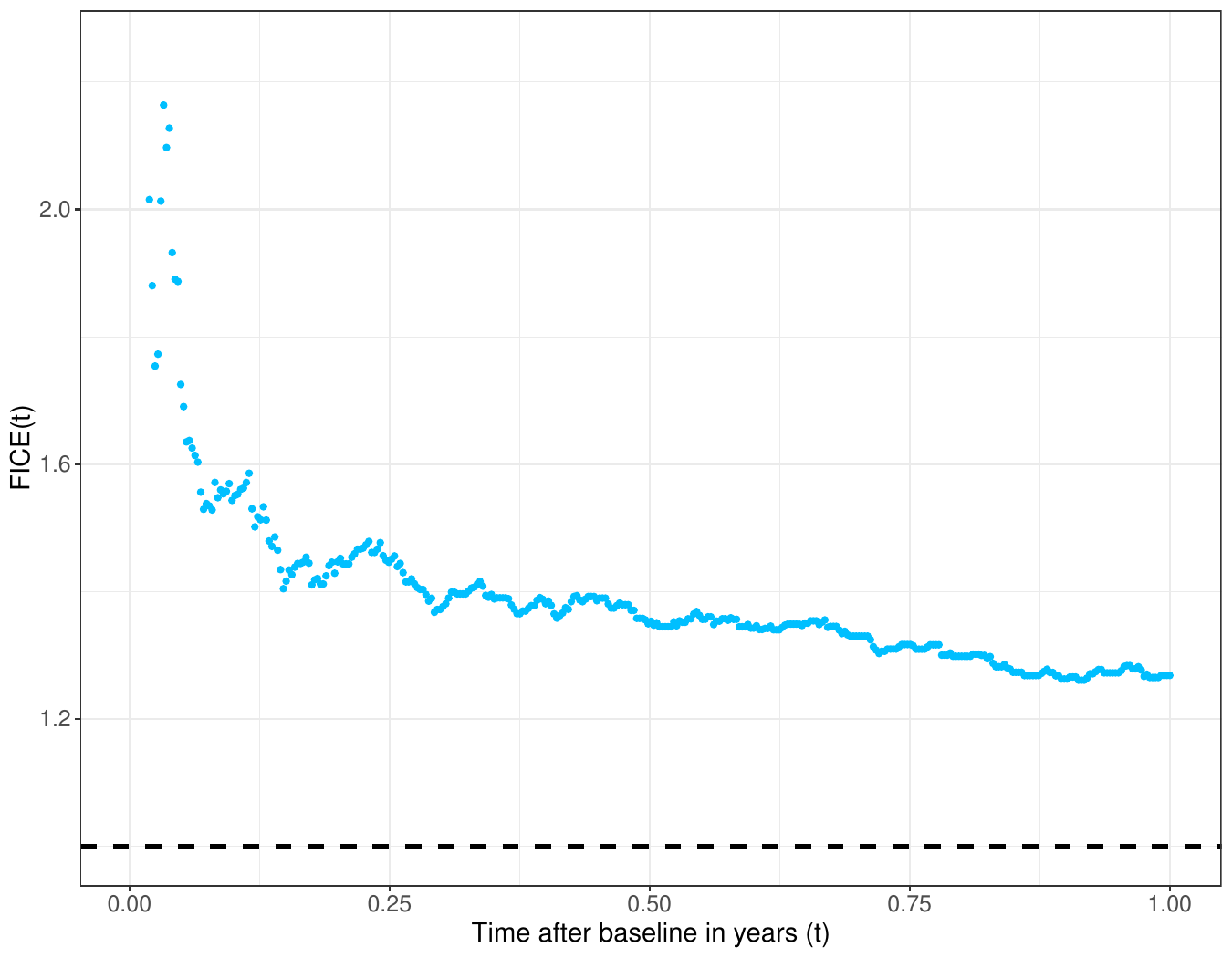}
\end{figure}

\subsection{Analysis under the frailty assumptions}
\label{Appsubsec:nalysis under the frailty assumptions}

\subsubsection{Regression coefficients}
\label{Appsubsec:coefficients}
Tables \ref{tab:cox_T2_A0_three_models} and \ref{tab:cox_T2_A1_three_models} below show the hazard ratio estimates obtained from the hazard models (Equation \eqref{conditional hazard models DGM}), with $A=0$ and $A=1$, respectively, employed for estimation of the causal effects under the frailty assumptions in Section \ref{Sec:Application} of the main text. 
Wald-type $95\%$ confidence intervals, with estimated SEs obtained by Bootstrap, are also reported. 
When including the sample location variables in the hazard models, their estimated coefficients' SEs were extremely high in transition $12$, presumably due to complete (or quasi) separation.
Therefore, we did not include these variables in the hazard models for transition $12$.

\begin{table}[H]
\centering
\footnotesize	
\caption{Cox regression model estimates for transitions $jk = 01, 02, 12$, within the 
Access group ($A=0$). HR: hazard ratio; CI95\%: 95\% confidence intervals. Cefta: ceftazidime.}
\label{tab:cox_T2_A0_three_models}
\begin{tabular}{|l|cc|cc|cc|}
\hline
\multirow{2}{*}{} 
  & \multicolumn{2}{c|}{\textbf{jk = 01}} 
  & \multicolumn{2}{c|}{\textbf{jk = 02}} 
  & \multicolumn{2}{c|}{\textbf{jk = 12}} \\
\cline{2-7}
\textbf{Covariate} & \textbf{HR} & \textbf{CI} & \textbf{HR} & \textbf{CI} & \textbf{HR} & \textbf{CI} \\
\hline
\textbf{Demographics} & & & & & & \\
Age & 1.01 & [0.92, 1.11] & 1.12 & [0.87, 1.44] & 1.44 & [0.24, 8.52] \\
Age$^2$ & 1.001 & [ 0.999, 1.003 ] & 0.999 & [ 0.996, 1.003 ] & 0.99 & [ 0.97, 1.02 ]  \\
Age$^3$ & 1.00 & [ 1.00, 1.00 ] & 1.00 & [ 1.00, 1.00 ] & 1.00 & [ 1.00, 1.00 ] \\ 
Male$^\dagger$ & 0.86 & [0.64, 1.15] & 0.80 & [0.64, 1.00] & 0.93 & [0.44, 1.98] \\
\hline
\textbf{Sample location at baseline} & & & & & & \\
Other & 1.00 & & 1.00 & & 1.00 & \\
Not taken & 1.67 & [0.07, 41.56] & 0.63 & [0.38, 1.04] &  &  \\
Multiple sources & 2.66 & [0.10, 69.37] & 0.78 & [0.44, 1.38] &  &  \\
Urine & 1.83 & [0.07, 47.38] & 0.86 & [0.50, 1.47] &  &  \\
Wound & 1.24 & [0.02, 62.67] & 0.53 & [0.23, 1.24] &  &  \\
Blood & 1.95 & [0.01, 307.51] & 1.80 & [0.96, 3.36] &  &  \\
Sputum & 1.35 & [0.02, 75.45] & 0.79 & [0.38, 1.63] &  &  \\
\hline
\textbf{Arrived from} & & & & & & \\
Other & 1.00 & & 1.00 & & 1.00 & \\
Home & 0.65 & [0.41, 1.04] & 0.90 & [0.62, 1.31] & 0.54 & [0.15, 1.90] \\
Institution & 1.26 & [0.72, 2.19] & 1.59 & [1.06, 2.39] & 0.69 & [0.16, 3.09] \\
\hline
\textbf{Hospitalization unit} & & & & & & \\
Other & 1.00 & & 1.00 & & 1.00 & \\
Internal & 1.18 & [0.68, 2.06] & 1.87 & [1.26, 2.76] & 5.46 & [1.99, 14.98] \\
Surgical & 0.53 & [0.29, 0.94] & 0.33 & [0.22, 0.50] & 0.54 & [0.03, 9.93] \\
\hline
\textbf{Medical history} & & & & & & \\
Dementia & 1.18 & [0.77, 1.80] & 0.84 & [0.61, 1.15] & 1.11 & [0.39, 3.13] \\
CRF & 1.25 & [0.81, 1.91] & 0.74 & [0.55, 0.99] & 1.49 & [0.61, 3.64] \\
Immunosuppression & 1.08 & [0.69, 1.69] & 1.68 & [1.25, 2.26] & 1.38 & [0.45, 4.20] \\
Diabetes & 1.45 & [1.07, 1.96] & 1.10 & [0.88, 1.38] & 0.90 & [0.44, 1.85] \\
Catheter & 2.85 & [1.82, 4.46] & 2.61 & [1.94, 3.52] & 1.78 & [0.79, 3.98] \\
Previous antibiotic (any) & 1.23 & [0.89, 1.71] & 1.07 & [0.83, 1.38] & 2.07 & [0.85, 5.04] \\
Previous cefta culture (365 days)$^\dagger$ & 1.65 & [1.08, 2.53] & 2.01 & [1.42, 2.83] & 0.29 & [0.08, 1.10] \\
\hline
\textbf{Medical information} & & & & & & \\
Arrival to culture ($>$ 2 days)$^\dagger$ & 2.59 & [1.75, 3.82] & 0.69 & [0.52, 0.91] & 0.75 & [0.33, 1.70] \\
Arrival to treatment ($>$ 2 days)$^\dagger$ & 1.05 & [0.68, 1.62] & 1.15 & [0.82, 1.60] & 1.49 & [0.64, 3.46] \\
\hline
\end{tabular}
\begin{flushleft}
$^\dagger$ \footnotesize Male is an indicator variable denoting whether the patient is a male. Arrival to culture and Arrival to treatment are indicators for whether arrival preceded culture collection or antibiotic treatment initiation, respectively, by more than two days. Previous ceftazidime culture is an indicator variable denoting whether a ceftazidime culture was taken before baseline.
\end{flushleft}
\end{table}

\begin{table}[H]
\centering
\footnotesize	
\caption{Cox regression model estimates for transitions $jk = 01, 02, 12$, within the 
Watch group (A=1). HR: hazard ratio; CI95\%: 95\% confidence intervals. Cefta: ceftazidime.}
\label{tab:cox_T2_A1_three_models}
\begin{tabular}{|l|cc|cc|cc|}
\hline
\multirow{2}{*}{} 
  & \multicolumn{2}{c|}{\textbf{jk = 01}} 
  & \multicolumn{2}{c|}{\textbf{jk = 02}} 
  & \multicolumn{2}{c|}{\textbf{jk = 12}} \\
\cline{2-7}
\textbf{Covariate} & \textbf{HR} & \textbf{CI} & \textbf{HR} & \textbf{CI} & \textbf{HR} & \textbf{CI} \\
\hline
\textbf{Demographics} & & & & & & \\
Age & 1.05 & [0.98, 1.12] & 1.41 & [1.08, 1.84] & 1.10 & [0.66, 1.84] \\
Age$^2$ & 1.00 & [ 0.999, 1.001 ] & 0.996 & [ 0.992, 0.999 ] & 0.999 & [ 0.991 1.001 ] \\ 
Age$^3$ & 1.00 & [ 1.00, 1.00 ] & 1.00 & [ 1.00, 1.00 ] & 1.00 & [ 1.00, 1.00 ] \\ 
Male$^\dagger$ & 1.09 & [0.87, 1.37] & 0.82 & [0.68, 0.98] & 0.90 & [0.52, 1.56] \\
\hline
\textbf{Sample location at baseline} & & & & & & \\
Other & 1.00 & & 1.00 & & 1.00 & \\
Not taken & 1.18 & [0.55, 2.51] & 0.40 & [0.26, 0.62] &  &  \\
Multiple sources & 1.37 & [0.57, 3.28] & 0.48 & [0.30, 0.77] &  &  \\
Urine & 1.67 & [0.79, 3.53] & 0.34 & [0.21, 0.55] &  &  \\
Wound & 1.06 & [0.43, 2.63] & 0.54 & [0.29, 0.99] &  &  \\
Blood & 2.34 & [0.97, 5.62] & 0.67 & [0.36, 1.23] &  &  \\
Sputum & 1.04 & [0.30, 3.60] & 0.49 & [0.25, 0.95] &  &  \\
\hline
\textbf{Arrived from} & & & & & & \\
Other & 1.00 & & 1.00 & & 1.00 & \\
Home & 0.65 & [0.44, 0.96] & 0.88 & [0.62, 1.24] & 1.47 & [0.52, 4.17] \\
Institution & 1.44 & [0.88, 2.36] & 1.25 & [0.85, 1.85] & 3.59 & [1.11, 11.57] \\
\hline
\textbf{Hospitalization unit} & & & & & & \\
Other & 1.00 & & 1.00 & & 1.00 & \\
Internal & 1.00 & [0.68, 1.46] & 0.69 & [0.51, 0.94] & 1.12 & [0.46, 2.72] \\
Surgical & 1.37 & [0.91, 2.08] & 0.21 & [0.14, 0.31] & 0.38 & [0.14, 1.02] \\
\hline
\textbf{Medical history} & & & & & & \\
Dementia & 1.23 & [0.86, 1.74] & 0.80 & [0.59, 1.09] & 0.40 & [0.18, 0.89] \\
CRF & 1.08 & [0.77, 1.51] & 1.04 & [0.76, 1.41] & 2.02 & [1.03, 3.93] \\
Immunosuppression & 1.58 & [1.14, 2.19] & 1.49 & [1.13, 1.96] & 2.07 & [1.08, 3.99] \\
Diabetes & 1.19 & [0.92, 1.53] & 0.81 & [0.65, 1.01] & 1.47 & [0.93, 2.33] \\
Catheter & 1.95 & [1.42, 2.68] & 2.11 & [1.61, 2.77] & 1.03 & [0.47, 2.27] \\
Previous antibiotic (any) & 1.33 & [1.03, 1.72] & 1.27 & [1.02, 1.58] & 1.51 & [0.75, 3.04] \\
Previous cefta culture (365 days)$^\dagger$ & 1.04 & [0.73, 1.49] & 1.46 & [1.10, 1.94] & 1.30 & [0.60, 2.79] \\
\hline
\textbf{Medical information} & & & & & & \\
Arrival to culture ($>$ 2 days)$^\dagger$ & 3.23 & [2.44, 4.28] & 0.80 & [0.62, 1.03] & 1.91 & [1.04, 3.52] \\
Arrival to treatment ($>$ 2 days)$^\dagger$ & 0.83 & [0.62, 1.10] & 0.93 & [0.72, 1.19] & 0.93 & [0.46, 1.89] \\
\hline
\end{tabular}
\begin{flushleft}
$^\dagger$ \footnotesize Male is an indicator variable denoting whether the patient is a male. Arrival to culture and Arrival to treatment are indicators for whether arrival preceded culture collection or antibiotic treatment initiation, respectively, by more than two days. Previous ceftazidime culture is an indicator variable denoting whether a ceftazidime culture was taken before baseline.
\end{flushleft}
\end{table}

\subsubsection{Data analysis under different $\rho$ values}
\label{Appsubsec:Data analysis with different rho values}
Figure \ref{Fig:res_frailty_all_diff_and_ratio_rho_values_0.5_and_1} depicts the estimated causal effects under the frailty assumptions for $\rho$ values different than zero. The results did not differ substantially from those presented in Section \ref{Sec:Application} of the main text, although the FICE$(t)$  remained non-significantly different from zero for a more extended period.

\begin{figure}[H]
    \centering
\caption{\footnotesize{Comparison between estimates of the different estimands within one year after baseline under the frailty assumptions with $\rho=0.5$ (first figure) and $\rho=1$ (second figure). 
The upper panel of each figure shows the various estimates on the difference scale. The lower panel shows the various estimates on the risk-ratio scale.}}
\label{Fig:res_frailty_all_diff_and_ratio_rho_values_0.5_and_1}
\includegraphics[width=0.7\textwidth]{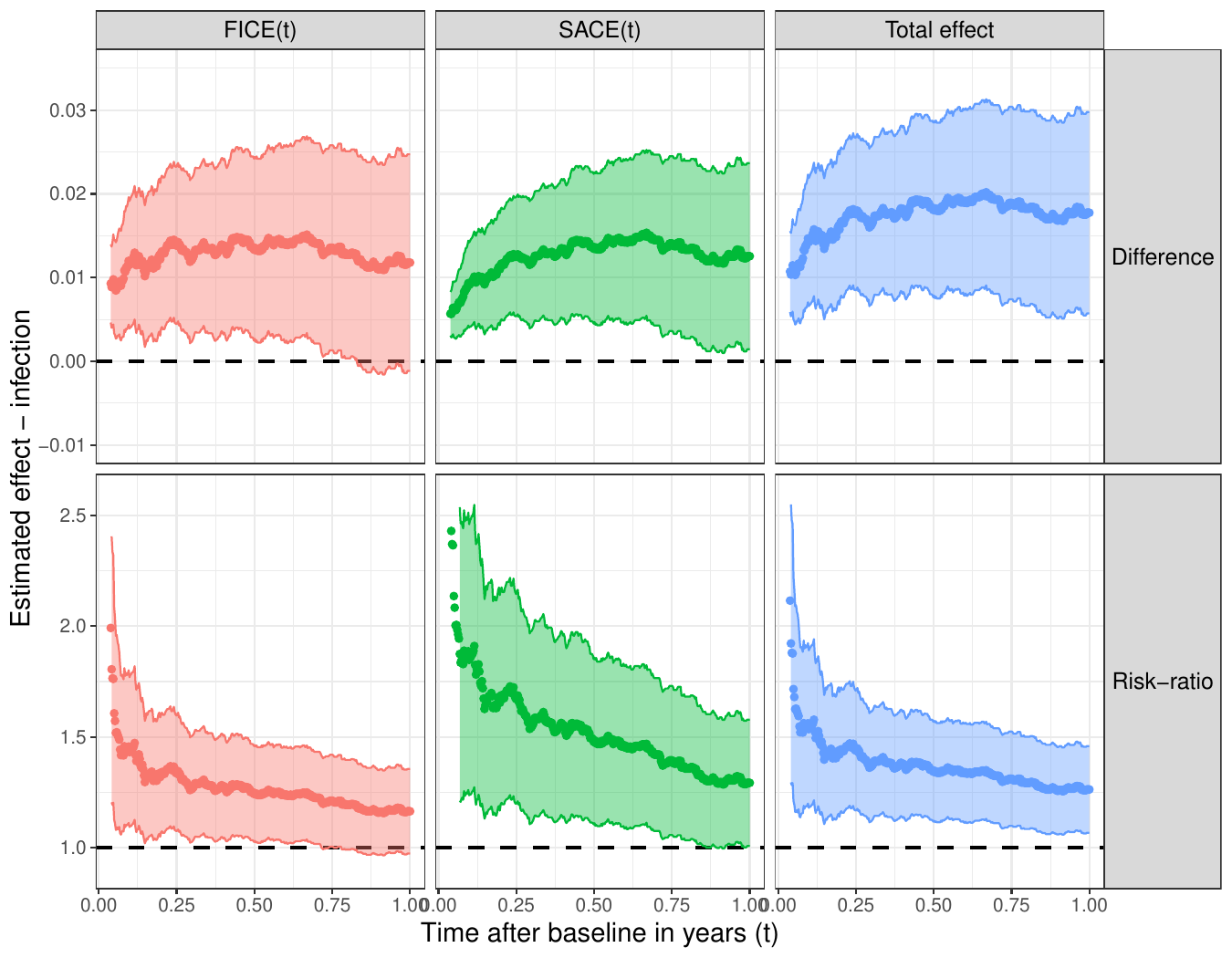}
    
\vspace{1em} 

\includegraphics[width=0.7\textwidth]{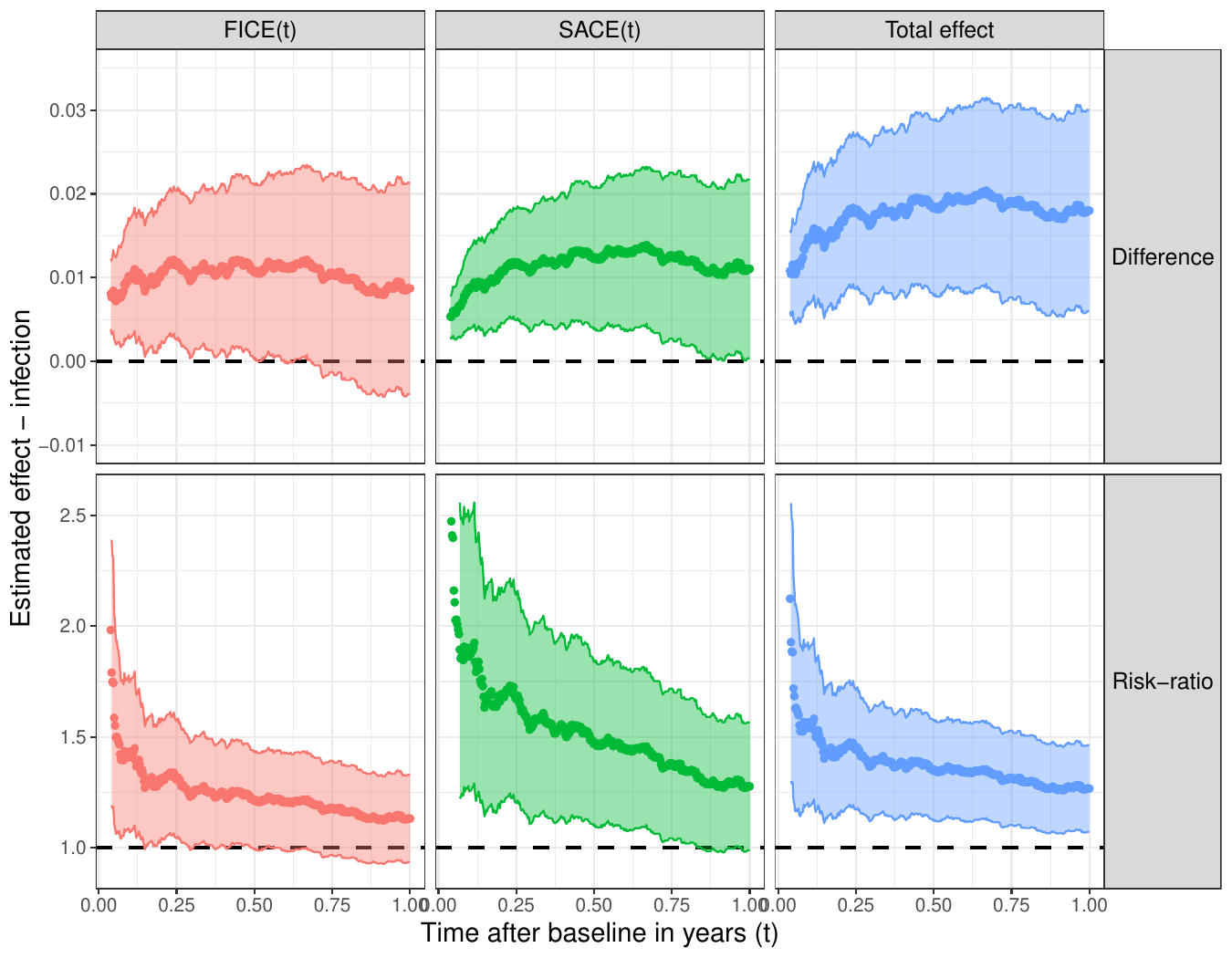}
    
    \label{fig:frailty_several_rho_values}
\end{figure}

\end{document}